\documentclass{llncs}

\usepackage{times}
\usepackage{amsmath}
\usepackage{amsfonts}
\usepackage{amssymb}
\usepackage{graphicx}
\usepackage{caption}
\usepackage{subfloat}
\usepackage{subfig}
\usepackage{tabularx}
\usepackage{multicol}
 \usepackage{stmaryrd}
 \usepackage{textcomp}

\usepackage[english]{babel}

\usepackage{macros_art}

\usepackage{tikz}
\usetikzlibrary{backgrounds,automata}

\usetikzlibrary{arrows}
\usetikzlibrary{shapes,shapes.geometric,arrows,fit,calc,positioning,automata}

\usepackage[toc,page]{appendix}

\DeclareGraphicsExtensions{.pdf,.png,.jpg}

\def\TTT{\mathcal{T}}
\def\OOO{\mathcal{O}}
\def\RRR{\mathcal{R}}
\def\KKK{\mathbb{K}}
\def\KLTL{KLTL }
\def\Acacia-K{Acacia-K }

\newcommand{\Prop}{\mathcal{P}}
\newcommand{\Propv}{\mathcal{P}_{v}}
\newcommand{\Propi}{\mathcal{P}_{i}}
\newcommand{\modelEnv}{\mathcal{M}_{E}}
\newcommand{\traces}{\text{traces}}
\newcommand{\trace}{\text{trace}}
\newcommand{\post}{\text{post}}
\newcommand{\tracev}{\text{trace}_{v}}
\newcommand{\proj}{\text{proj}}
\newcommand{\exec}{\textnormal{\text{exec}}}

\newcommand{\lang}{\mathcal{L}}

\newcommand{\aut}{\mathcal{A}}

\newcommand{\obs}{\text{o}}

\newcommand{\pKLTL}{\ensuremath{\textnormal{KLTL}^+}}

\newcommand{\RX}[2]{\textnormal{\text{exec}}(\modelEnv[#1], #2)}
\newcommand{\R}[1]{\textnormal{\text{exec}}(\modelEnv, #1)}

\begin{document}

\sloppy
\title{Safraless Synthesis for Epistemic Temporal Specifications\thanks{Work partially supported by the ANR research project ``EQINOCS'' no. ANR-11-BS02-0004}}

\author{Rodica Bozianu\inst{1} \and C\u at\u alin Dima\inst{1} \and Emmanuel Filiot\inst{2,}\thanks{F.R.S.-FNRS research associate (chercheur qualifi\'{e})}}


\institute{Universit\'{e} Paris Est, LACL (EA 4219), UPEC, 94010 Cr\'{e}teil Cedex, France\and
Universit\'{e} Libre de Bruxelles, CP 212 - 1050 Bruxelles Belgium 
}

\maketitle

\vspace*{-10pt}
\paragraph{\textbf{Abstract. }} 
In this paper we address the synthesis problem for specifications given in
linear temporal single-agent epistemic logic, KLTL (or $KL_1$), over single-agent systems having imperfect information of the environment state. \cite{MeVa07} have shown that this problem is 2Exptime complete.
However, their procedure relies on complex automata constructions that are
notoriously resistant to efficient implementations as they use Safra-like
determinization.

We propose a "Safraless" synthesis procedure for a large fragment of KLTL.
The construction transforms first the synthesis problem into the problem
of checking emptiness for universal co-B\"{u}chi tree automata using an
information-set construction.
Then we build a safety game that can be solved using an antichain-based 
symbolic technique exploiting the structure of the underlying automata.
The technique is implemented and applied to a couple of case studies.

\vspace*{-2mm}
\section{Introduction}

 \vspace{-2mm}

The goal of system verification is to check that a system 
satisfies a given property. One of the major achievements in
system verification is the theory of \emph{model checking}, that uses
automata-based techniques to check properties expressed in temporal
logics, for systems modelled as transitions systems. 
The \emph{synthesis} problem is more ambitious: given a
specification of the system, the aim is to automatically synthesise a
system that fulfils the constraints defined by the
specification. Therefore, the constraints do not  need to be checked a
posteriori, and this allows the designer to focus on defining high-level
specifications, rather than designing complex computational models of
the systems.

\emph{Reactive systems} are non-terminating systems that 
interact with some environment, e.g., hardware or software that
control transportations systems, or medical devices. One
of the main challenge of synthesis of reactive systems is to cope with the uncontrollable
behaviour of the environment, which usually leads to computationally
harder decision problems, compared to system verification. For
instance, model-checking properties expressed in \emph{linear time
  temporal logic} (LTL) is \textsf{PSpace-c} while LTL synthesis is 
\textsf{2Exptime-c} \cite{PnRo89}. Synthesis of reactive systems from
temporal specifications has gain a lot of
interest recently as several works have shown its practical
feasibility \cite{KuVa05,BJPS11,Ratsy,JobstmannB06,Filiot11}. 
These progresses were supported by Kupferman and Vardi's
breakthrough in automata-based synthesis techniques \cite{KuVa05}. 
More precisely, they have shown that the complex Safra's
determinization operation, used in the classical LTL synthesis 
algorithm \cite{PnRo89}, could be avoided by working directly with
universal co-B\"uchi automata. Since then, several
other ``Safraless'' procedures have been defined
\cite{KuVa05,ScheweF07a,conf/formats/GiampaoloGRS10,Filiot11}.
In \cite{ScheweF07a,Filiot11}, it is shown that LTL synthesis reduces 
to testing the emptiness of a universal co-B\"uchi tree automaton, that
in turn can be reduced to solving a safety game. The structure of the
safety games can be exploited to define a symbolic game solving
algorithm based on compact antichain representations \cite{Filiot11}.

In these works, the system is assumed to have perfect
information about the state of the environment. However in many practical
scenarios, this assumption is not realistic since some 
environment information may be hidden to the system (e.g. private
variables). Towards the (more realistic) synthesis of partially informed
systems, imperfect information two-player games on graphs 
have been studied
\cite{journals/lmcs/RaskinCDH07,DBLP:conf/vmcai/Chatterjee0FR14,conf/fsttcs/BerwangerD08,DBLP:journals/fmsd/Chatterjee0H13}. However,
they consider explicit state transition systems rather than synthesis
from temporal specifications. Moreover, the winning objectives that
they consider cannot express fine properties about imperfect
information, i.e., cannot speak about knowledge.

 \vspace{-2mm}
\paragraph{Epistemic Temporal Logics} \cite{HaYo84} are logics formatted for reasoning about multi-agent situations.
They are extensions of temporal logics 
with knowledge operators $K_i$ for each agent.
They have been successfully used for verification of various distributed 
systems in which the knowledge of the agents is essential for the correctness of the system specification.

 \vspace{-2mm}
\paragraph{Synthesis problem with temporal epistemic objectives} 
 Vardi and van der Meyden \cite{MeVa07} have considered
epistemic temporal logics  to define specifications that
can, in addition to temporal properties, also express properties that
refer to the imperfect information, and they studied the synthesis
problem. They define the synthesis problem in a
multi-agent setting, for specifications written in LTL extended with knowledge
operators $K_i$ for each agent (KLTL).
In such models, transitions between states of the
environment model depend on actions of the environment and the
system. The system does not see which actions are played by the
environment but get some observation on the states in which 
the environment can be (observations are subsets of states). 
An execution of the environment model, from the
point of view of the system, is therefore an infinite sequence
alternating between its own actions and observations. 

The goal of the
KLTL synthesis problem is to automatically generate a strategy for the
system (if it exists) that tells it which action should be played,
depending on finite histories, so that whatever the environment does, 
all the (concrete) infinite executions resulting from this strategy
satisfy the KLTL formula. In \cite{MeVa07}, this problem was shown to
be undecidable even for two agents against
the environment.
On the other hand, for  single
agent against environment situations, they show that the problem is  \textsf{2Exptime-c},
by reduction to
the emptiness of alternating B\"uchi automata. This theoretically
elegant construction is however difficult to implement and optimize,
as it relies on complex Safra-like automata operations (Muller-Schupp
construction).

\vspace{-2mm}
\paragraph{Contributions} In this paper, we follow the formalisation
of \cite{MeVa07} and, as our main contribution, define and implement a Safraless
synthesis procedure for the positive fragment of KLTL ($\pKLTL$),  i.e.,  KLTL
formulas where the operator $K$ does not occur under an odd number of 
negations. Our procedure relies on universal co-B\"uchi tree
automata (UCT). More precisely, given a $\pKLTL$ formula $\varphi$ and
some environment model $\modelEnv$, we show how to construct a UCT
$\TTT_\varphi$ whose language is exactly the set of strategies that
realize $\varphi$ in $\modelEnv$.

Despite the fact that our procedure has \textsf{2-ExpTime}
worst-case complexity, we have implemented it and shown its practical
feasibility through a set of examples. In particular, based on ideas
of \cite{Filiot11}, we reduce the problem of checking the emptiness of
$\TTT_\varphi$ to solving a safety game whose state space can be
ordered and compactly represented by antichains. Moreover, rather that
using the reduction of \cite{Filiot11} as a blackbox, we further
optimize the antichain representations to improve their compactness. 
Our implementation is based on the tool Acacia \cite{CAV12} and, to
the best of our knowledge, it is the first implementation of a
synthesis procedure for epistemic temporal specifications. 
As an application, this implementation can be used to solve two-player
games of imperfect information whose objectives are given as LTL
formulas, or universal co-B\"uchi automata.

\vspace{-2mm}
\paragraph{Organization of the paper} In Section \ref{sec:KLTLReal},
we define the KLTL synthesis problem. In Section \ref{sec:automata},
we define universal co-B\"uchi automata for infinite words and trees. 
In Section \ref{sec:ltl}, we consider the particular case of LTL
synthesis in an environment model with imperfect information. The
construction explained in that section will be used in the
generalization to \pKLTL and moreover, it can be used to solve
two-player imperfect information games with LTL (and more generally
$\omega$-regular) objectives. In Section \ref{sec:kltl}, we define our
Safraless procedure for \pKLTL, and show in Section
\ref{sec:antichains} how to implement it with antichain-based symbolic
techniques. Finally, we describe our implementation in Section \ref{sec:implementation}.
\emph{Full proofs can be found in Appendix in which, for self-containdness, we
also explain the reduction to safety games.}

\vspace*{-2mm}
\section{\KLTL Realizability and Synthesis}
\label{sec:KLTLReal}

In this section, we define the realizability and synthesis problems
for \KLTL specifications, for one partially informed agent, called the
\emph{system}, against an environment.

 \vspace{-3mm}
\paragraph{\textbf{Environment Model}} We 
assume to have, as input of the problem, a model of the behaviour of
the environment as a transition system. This transition system  is defined
over two disjoint sets of actions $\Sigma_1$ and $\Sigma_2$, for the
system and the environment respectively.  The transition relation from
states to states is defined with respect to pairs
of actions in $\Sigma_1\times \Sigma_2$. Additionally, each state $s$ of
the environment model carries an interpretation $\tau_e(s)$ over a
(finite) set of propositions $\Prop$. However, the system is not perfectly informed
about the value of some propositions, i.e., 
some propositions are visible to the system, and some are
not. Therefore, we partition the set $\Prop$ into 
two sets $\Propv$ (the visible propositions) and $\Propi$ (the
invisible ones).

An \textit{environment model}  is a tuple $\modelEnv =
(\Prop,\Sigma_1,\Sigma_2,S_e, S_0, \Delta_e, \tau_e)$ where
\vspace{-1mm}
\begin{itemize}
\item $\Prop$ is a finite set of propositions, $\Sigma_1$ and $\Sigma_2$ are finite set of actions for the
  system and the environment resp., 
\item $S_e$ is a set of states, $S_0\subseteq S_e$ a set of
  initial states,
\item $\tau_e : S_e \rightarrow 2^{\Prop}$ is a labelling function,
\item $\Delta_e \subseteq S_e \times \Sigma_1 \times \Sigma_2 \times
 S_e$ is a transition relation.
\end{itemize}
\vspace{-1mm}

The model is assumed to be deadlock-free, i.e. from any state, there
exists at least one outgoing transition. Moreover, the model is
assumed to be complete for all actions of the system, i.e. for all
states and all actions of the system, there exists an outgoing
transition. %
The set of \emph{executions} of $\modelEnv$, denoted by
$\exec(\modelEnv)$, is the set of infinite sequences of states
$\rho = s_0s_1\dots \in S_e^\omega$ such that $s_0\in S_0$ and for all $i>0$,
$(s_i,a_1,a_2,s_{i+1})\in \Delta_e$ for some
$(a_1,a_2)\in\Sigma_1\times \Sigma_2$. Given a sequence of
states $\rho = s_0s_1\dots$ and a set $P\subseteq \Prop$, we
denote by $\trace_P(\rho)$ its projection over $P$, i.e.
$\trace_P(\rho) = (\tau_e(s_0)\cap P)(\tau_e(s_1)\cap P)\dots$. 
The \emph{visible trace} of $\rho$ is defined by $\tracev(\rho) =
\trace_{\Propv}(\rho)$. 
The \emph{language} of $\modelEnv$ with respect to $P$ is defined as 
$\lang_P(\modelEnv) = \{ \trace_P(\rho)\ |\ \rho\in \exec(\modelEnv)\}$. 
The language of $\modelEnv$  is defined as $\lang_\Prop(\modelEnv)$. The visible language of $\modelEnv$ is defined as
 $\lang_{\Propv}(\modelEnv)$.
Finally, given an infinite sequence of actions $a = a_1^0a_2^0\dots \in
(\Sigma_1.\Sigma_2)^\omega$ and an execution $\rho = s_0s_1\dots$ of
$\modelEnv$, we say that $a$ is compatible with $\rho$ if for all
$i\geq 0$, $(s_i, a_1^i, a_2^i, s_{i+1})\in \Delta_e$.

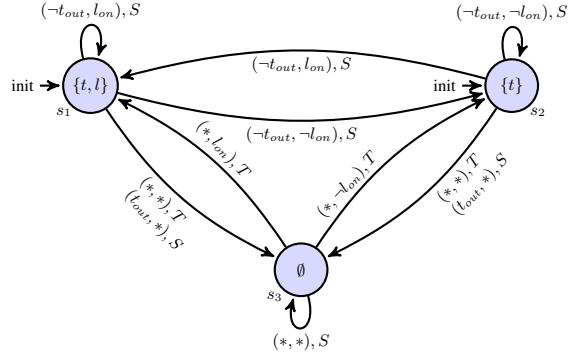
\begin{figure}[t]
\centering



\begin{center}
\begin{tikzpicture}[->,>=stealth',shorten >=1pt,auto,node distance=2.2cm,
                    thick,scale=0.7,every node/.style={scale=0.7}]
  \tikzstyle{every
    state}=[shape=circle,fill=blue!15,text=black,minimum size=1cm]
  \tikzstyle{every edge}=[draw=black]
  \tikzstyle{initial}=[initial by arrow, initial where=left, initial text=init]

  \node[initial,state] (s1) at (0,0) {$\{t,l\}$};
  \node[state] (s3)  at (4,-3.5) {$\emptyset$};
  \node[initial,state,initial where=right] (s2)  at (8,0) {$\{t\}$};

  \node[state,draw=none,fill=none] (s1l) at (-0.5,-0.5) {$s_1$};
  \node[state,draw=none,fill=none] (s3l)  at (3.5,-4) {$s_3$};
  \node[state,draw=none,fill=none] (s2l)  at (8.5,-0.5) {$s_2$};

\path (s1) [loop above] edge node {$(\neg t_{out}, l_{on}),S$} (s1) ;
\path (s2) [loop above] edge node {$(\neg t_{out}, \neg l_{on}),S$} (s2) ;
\path (s3) [loop below] edge node {$(*, *),S$} (s3) ;

\path (s2) [bend right=15] edge node {$(\neg t_{out}, l_{on}),S$} (s1) ;
\path (s1) [bend right=15] edge node[below]{$(\neg t_{out}, \neg l_{on}),S$} (s2) ;

\path (s1) [bend right=15] edge node[rotate=-48, below] {$\begin{array}{ll}(*,
      *),T\\ (t_{out},*),S\end{array}$} (s3) ;
\path (s3) [bend right=15] edge node[rotate=-48, above] {$(*, l_{on}),T$} (s1) ;

\path (s2) [bend left=15] edge node[rotate=48] {$\begin{array}{ll}(*,
      *),T\\ (t_{out},*),S\end{array}$} (s3) ;
\path (s3) [bend left=15] edge node[rotate=48] {$(*, \neg l_{on}),T$} (s2) ;

  

  

  



\vspace{-10mm}
\end{tikzpicture}
\end{center}
\vspace{-6mm}
\caption{Environment model $\modelEnv$ of Example \ref{ex:light}}
\label{fig:env}
\vspace*{-6mm}
\end{figure}


This formalization is very close to that of \cite{MeVa07}. However in
\cite{MeVa07}, partial observation is modeled as a partition of the
state space. The two models are equivalent. In particular, we will
see that partitioning the propositions into visible and invisible ones also
induces a partition of the state space into observations.

\begin{example}\label{ex:light}
We illustrate the notion of environment model on the example of
\cite{MeVa07}, that describes the behaviour of an environment against
a system acting on a timed toggle switch with two positions (on,off)
and a light. It is depicted in Fig. \ref{fig:env}. %
The set $\Prop =\{ t,l\}$ contains two propositions $t$ (true iff the
toggle is on) and $l$ (true iff the light is on). Actions of the
system are $\Sigma_1 = \{ T, S\}$ for ``toggle'' and ``skip''
respectively.  The system can change the position of the toggle only
if it plays $T$, and $S$ has no effect. Actions of the environment are 
$\Sigma_2 = \{  (t_{out}, l_{on})\ |\ t_{out}, l_{on}\in
\{0,1\} \}$. The boolean variables $(t_{out}$ and $l_{on})$ indicate that 
the environment times out the toggle and that it switches on the
light. The transition function is depicted on the figure as well as
the labelling function $\tau_e : S_e\rightarrow 2^\Prop$. The star $*$
means ``any action''. The light can be on only if the
toggle in on (state $s_1$), but it can be off even if the toggle is on
(state $s_2$), in case it is broken. This parameter is uncontrollable
by the system,  and therefore it is
controlled by the environment (action $l_{on}$). The timer is assumed to be unreliable
and therefore the environment can timeout at any time (action
$t_{out}$). The system sees only the light, i.e. $\Propv = \{ l\}$ and
$\Propi = \{ t\}$. The goal of the system is to have a strategy such
that he always knows the position of the toggle. 

\end{example}
 
\vspace*{-2mm}
\paragraph{\textbf{Observations.}} 
The partition of the set of propositions
$\Prop$ into a set of visible propositions $\Propv$ and a set of
invisible propositions $\Propi$ induces an indistinguishability relation
over the states $S_e$. Two states are \textit{indistinguishable}, denoted $s_1 \sim s_2$, if
they have the same visible propositions, i.e. $\tau_e(s_1)\cap \Propv
= \tau_e(s_2)\cap \Propv$.
It is easy to see that $\sim$ is an equivalence relation over
$S_e$. Each equivalence class of $S_e$ induced by $\sim$ is called an
\emph{observation}. The equivalence class of a state $s\in S_e$ is
denoted by $\obs(s)$ and the set of observations is denoted by $\OOO$.
The relation $\sim$ is naturally extended to (finite or infinite) executions: $\rho_1\sim
\rho_2$ if $\tracev(\rho_1)=\tracev(\rho_2)$. %
%
%
Similarly, two executions $\rho = s_0s_1\dots$ and $\rho'=s'_0s'_1\dots$ are said to be
indistinguishable \emph{up to some position $i$} if $\tracev(s_0\dots
s_i) = \tracev(s'_0\dots s_i')$. This indistinguishability notion
also is an equivalence relation over executions that we denote by
$\sim_i$. 

Coming back to Example \ref{ex:light}, since the set of visible propositions is $\Propv = \{l\}$ and  the set of invisible ones is $\Propi=\{t\}$, the states $s_2$ and $s_3$ are indistinguishable (in both $s_2$, $s_3$ the light is off) and therefore $\OOO = \{ o_0, o_1 \}$ with $o_0 = \{s_2, s_3 \}$ and $o_1 = \{s_1\}$.


Given an infinite sequence $u = a_1o_1a_2o_2\dots \in
(\Sigma_1.\OOO)^\omega$ of actions of Player $1$, and observations,
we associate with $u$ the set of possible executions of $\modelEnv$
that are compatible with
$u$. Formally, we define $\exec(\modelEnv, u)$ the set of executions $\rho =
s_0s_1\dots \in \exec(\modelEnv)$ such that for all $i\geq 1$, 
$\obs(s_i) = o_i$ and there exists an action $b_i$ of the environment
such that $(s_{i-1}, a_i, b_i, s_i)\in \Delta_e$. We also define the
traces of $u$ as the set of traces of all executions of $\modelEnv$
compatible with $u$, i.e. $\traces(u) = \{ \traces(\rho)\ |\ \rho\in
\exec(\modelEnv, u)\}$.

\vspace*{-2mm}
\paragraph{\textbf{Epistemic Linear Time Temporal Logic (KLTL)}} We
now define the logic KLTL for one-agent (the system).
 The logic \KLTL extends the logic LTL with an epistemic operator $K\phi$, modelling the
property that the system knows that the formula $\phi$ holds. \textit{KLTL formulae} are defined over the set of atomic propositions $\Prop$ by:
$$\varphi ::= p \hspace*{7pt}|\hspace*{7pt} \neg \varphi \hspace*{7pt}|\hspace*{7pt} \varphi \vee \varphi \hspace*{7pt}|\hspace*{7pt} \bigcirc \varphi \hspace*{7pt}|\hspace*{7pt} \varphi_1 \mathcal{U} \varphi_2 \hspace*{7pt}|\hspace*{7pt} K \varphi \hspace*{7pt} \vspace*{-2mm}$$
in which $p \in \Prop$ and $\bigcirc$ and $\mathcal{U}$ are the "next" and "until" operators from linear temporal logic.
Formulas of the type $K \varphi$ are read as "the system knows that $\varphi$ holds". 
 We define the macros $\Diamond$ (eventually)
 and $\Box$ (always) as usual. %
LTL is the fragment
of KLTL without the $K$ operator.


The semantics of a KLTL formula $\varphi$ is defined for an environment model
$\modelEnv = (\Prop, \Sigma_1,\Sigma_2,S_e,S_0,\Delta_e,\tau_e)$, a set of
executions $R\subseteq \exec(\modelEnv)$, an
execution $\rho=s_0s_1\dots\in R$ and a position $i\geq 0$ in $\rho$. 
It is defined inductively:
\vspace{-2mm}
\begin{itemize}
\item  $R,\rho,i \models p$ if $p \in \tau_e(s_i)$,
\item  $R,\rho,i \models \neg \varphi$ if $R,\rho,i \not \models \varphi$,
\item  $R,\rho,i \models \varphi_1 \vee \varphi_2$ if $R,\rho,i \models \varphi_1$ or $R,\rho,i \models \varphi_2$,
\item  $R,\rho,i \models \bigcirc \varphi$ if $R,\rho,i+1 \models \varphi$,
\item  $R,\rho,i \models \varphi_1 \mathcal{U} \varphi_2$ if $\exists j \geq i$ s.t. $R,\rho,j \models \varphi_2$ and $\forall i \leq k < j,$  $R,\rho,k \models \varphi_1$,
\item  $R,\rho,i \models K \varphi$ if for all
  $\rho'\in R$ s.t. $\rho\sim_{i}\rho'$, we have $R,\rho',i \models \varphi$.
\end{itemize}
\vspace{-2mm}

In particular, the system knows $\varphi$ at position $i$ in the execution
$\rho$, if all other executions in $R$ whose prefix up to position $i$ are
indistinguishable from that of $\rho$, also satisfy $\varphi$. 
We write $R, \rho\models \varphi$ if $R, \rho,0\models
\varphi$, and $R\models \varphi$ if $R, \rho\models \varphi$
for all executions $\rho\in R$. We also write $\modelEnv\models \varphi$ 
to mean $\exec(\modelEnv)\models \varphi$. 
Note that $\modelEnv\models \varphi$ iff $\modelEnv\models K\varphi$.

Consider Example \ref{ex:light} and the set $R$ of executions that
eventually loops is $s_1$. Pick any $\rho$ in $R$. 
Then $R,\rho,0\models \Box K \Diamond(l)$. Indeed, take any position $i$ in
$\rho$ and any other executions $\rho'\in R$ such that $\rho\sim_i
\rho'$. Then since $\rho'$ will eventually loop in $s_1$, it will
satisfy $\Diamond(l)$. Therefore $R,\rho,i\models K\Diamond(l)$, for all $i\geq 0$.

\vspace*{-2mm}
\paragraph{\textbf{KLTL Realizability and Synthesis}} As presented in
\cite{Filiot11} for the perfect information setting, 
the realizability problem, given the environment model $\modelEnv$ and the KLTL formula $\varphi$,
 is best seen as a turn-based game between the system
(Player $1$) and the environment (Player $2$). 
In the first round of the play, Player $1$ picks some action $a_1^0 \in
 \Sigma_1$ and then Player $2$ picks some action in $a_2^0\in
 \Sigma_2$ and solves the nondeterminism in $\Delta_e$, and a new round starts. The two players play for an
 infinite duration and the outcome is an infinite sequence
 $w = a_1^0a_2^0a_1^1a_2^1\dots$. The winning objective is given by
 some \KLTL formula $\varphi$. Player $1$ wins the play if 
 for all executions $\rho$ of $\modelEnv$ that are compatible with $w$, we have 
$\modelEnv,\rho\models \varphi$.

Player $1$ plays according to \emph{strategies} (called
\emph{protocols} in \cite{MeVa07}). Since Player $1$ has only partial
information about the state of the environment, his strategies are based on
the histories of his own actions and the observations he got from the
environment. Formally, a strategy for Player $1$ is a mapping
$\lambda:(\Sigma_1\OOO)^*\rightarrow \Sigma_1$, where as defined
before, $\OOO$ denotes the set of observations of Player $1$ over the
states of $\modelEnv$. Fixing a strategy $\lambda$ of Player $1$
restricts the set of executions of the environment model $\modelEnv$. 
An execution $\rho=s_0s_1\dots \in\exec(\modelEnv)$ is said to be \emph{compatible} with
$\lambda$ if there exists an infinite sequence of actions $a =
a_1^0a_2^0...\in(\Sigma_1.\Sigma_2)^\omega$, compatible with $\rho$
such that for all $i\geq 0$, $a_1^i =
\lambda(a_1^0\obs(s_0)a_1^1\obs(s_1)... a_1^{i-1} \obs(s_{i-1}))$. 
We denote by $\exec(\modelEnv, \lambda)$ the set of
executions of $\modelEnv$ compatible with $\lambda$.




\begin{definition}
A \KLTL formula $\varphi$ is \textit{realizable in $\modelEnv$} if
there exists a strategy $\lambda$ for the system such that
$\exec(\modelEnv, \lambda)\models \phi$.

\end{definition}

\begin{theorem}[\cite{MeVa07}]
    The \KLTL realizability problem (for one agent) is \textsf{2ExpTime-complete}.
\end{theorem}

If a formula is realizable, the \emph{synthesis} problem asks to
generate a finite-memory strategy that realizes the formula. Such a
strategy always exists if the specification is realizable \cite{MeVa07}.
Finite memory strategies can be represented by Moore machines that
read observations and output actions of Player $1$. We refer the
reader to \cite{Filiot11} for a formal definition of finite-memory
strategies.

Considering again Example \ref{ex:light}, the formula 
$\varphi = \Box(K(t) \vee K(\neg t))$ 
expresses the fact that the system knows at each step the position of the toggle.
As argued in \cite{MeVa07}, this formula is realizable if the initial
set of the environment is $\{ s_1, s_2 \}$ since both states are
labelled with $t$. Then, one strategy of the system is to play first time $T$, action that will lead to $s_3$,
and then always play $S$ in order to stay in that state. Following this strategy,
in the first step the formula $K(t)$ is satisfied and then $K(\neg t)$ becomes true.
However, the formula is not realizable if the set of initial
states of the environment is $\{ s_2, s_3 \}$ since from the beginning
the system doesn't know the value of the toggle. 

\vspace*{-2mm}
\section{Automata for Infinite Words and Trees}\label{sec:automata}

\vspace*{-2mm}
\paragraph{\textbf{Automata on infinite Words}} An \textit{infinite word automaton} over some (finite) alphabet $\Sigma$ is a tuple $A=(\Sigma, Q, Q_0, \Delta, \alpha )$ where $\Sigma$ is the finite input alphabet, $Q$ is the finite set of states, 
$Q_0 \subseteq Q$ is the set of initial states, 
$\alpha \subseteq Q$ is the set of final states (accepting states) 
and $\Delta \subseteq Q \times \Sigma \times Q$ is the transition relation.

For all $q \in Q$ and all $ \sigma \in \Sigma$, we let
$\Delta(q,\sigma) = \{ q' | (q,\sigma, q') \in \Delta \}$. We let $|A|
= |Q|+|\Delta|$. We say that $A$ is \textit{deterministic} if $|Q_0| =
1$ and  $\forall q \in Q, \forall \sigma \in \Sigma,
|\Delta(q,\sigma)| \leq 1$. It is \textit{complete} if $\forall q \in
Q, \forall \sigma \in \Sigma, \Delta(q,\sigma) \not = \emptyset$. In
this paper we assume, w.l.o.g., that the word automata are always complete.

A \textit{run} on the automaton $A$ over an infinite input word $w =
w_0 w_1 w_2 ...$, is a sequence $r = q_0 q_1 q_2 ... \in Q^\omega$
such that $(q_i, w_i, q_{i+1}) \in \Delta$ for all $i \geq 0$ and $q_0
\in Q_0$. We denote by $Runs_A(w)$ the set of runs of $A$ on $w$ and
by $Visit(r,q)$ the number of times the state $q$ is visited along the
run $r$(or $\infty$ if the path visit the state $q$ infinitely
often). Here, we consider two accepting conditions for infinite word
automata and name the infinite word automata according to the used
accepting condition. Let $B \in \mathbb{N}$. A word $w \in \Sigma^\omega$ is accepted by $A$ if (according to the accepting condition): 
\begin{align*}
&\text{ Universal Co-B\"{u}chi} &: \forall r \in Runs_A(w), \forall q \in \alpha,\ Visit(r, q) < \infty \\
&\text{ Universal B-Co-B\"{u}chi} &: \forall r \in Runs_A(w), \forall q \in \alpha,\ Visit(r, q) \leq B 
\end{align*}   

The set of words accepted by $A$ with the universal co-B\"{u}chi
(resp. $B$-co-B\"{u}chi) accepting condition is denoted by
$\mathcal{L}_{uc}(A)$ (resp. $\mathcal{L}_{uc,B}(A)$) . We say that
$A$ is a \textit{universal co-B\"{u}chi word automaton} (UCW) if the first acceptance condition is used and that
$(A,B)$ is an \textit{universal B-co-B\"{u}chi word automaton} (UBCW) if the second one is used.

 Given an LTL formula $\varphi$, we can translate it into an equivalent universal co-B\"{u}chi word automaton $A_\varphi$ \label{coBu}. This can be done with a single exponential blow-up by first negating $\varphi$, then translating $\neg \varphi$ into an equivalent nondeterministic B\"{u}chi word automaton, and then 
dualize it into a universal co-B\"{u}chi word automaton
\cite{Filiot11,KuVa05}. 


\vspace*{-2mm}
\paragraph{\textbf{Automata on Infinite Trees}}

Given a finite set $D$ of directions, a $D{-}tree$ is a prefix-closed set $T \subseteq D^*$, i.e., if $x \cdot d \in T$, where $d \in D$, then $x \in T$. 
The elements of $T$ are called \textit{nodes} and the empty word
$\epsilon$ is the \textit{root} of $T$. For every $x \in T$, the nodes
$x \cdot d$, for $d \in D$, are the \textit{successors} of $x$. A node $x$ is a
\textit{leaf} if has no successor in $T$, formally, $\forall d \in D,
x \cdot d \not \in T$.  The tree $T$ is \textit{complete} if for all
nodes, there are successors in all directions, formally, $\forall x
\in T, \forall d \in D, x \cdot d \in T$. 
Finite and infinite branches $\pi$ in a tree $T$ are naturally defined,
respectively, as finite and infinite paths in $T$ starting from the
root node. %
%
%
Given an alphabet $\Sigma$, a \emph{$\Sigma {-}$labelled $D{-}$tree} is a
pair $\langle T, \tau \rangle$ where $T$ is a tree and $\tau : T
\rightarrow \Sigma$ maps each node of $T$ to a letter in
$\Sigma$. We omit $\tau$ when it is clear from the context. 
Then, in a tree $T$, an infinite (resp. finite) branch $\pi$
induces an infinite (resp. finite) sequence of labels and directions
in $(\Sigma.D)^\omega$ (resp. $(\Sigma.D)^*\Sigma$). 
We denote this sequence by $\tau(\pi)$. For instance, for a set of
system's actions $\Sigma_1$ and a set of observations $\OOO$, 
 a strategy $\lambda:(\Sigma_1\OOO)^*\rightarrow \Sigma_1$ of the
 system can be seen has a $\Sigma_1$-labelled $\OOO$-tree whose nodes
 are finite outcomes\footnote{Technically, a strategy $\lambda$ is
   defined also for histories that are not accessible by $\lambda$
   itself from the initial (empty) history $\epsilon$. The tree represents only
   accessible histories but we can, in the rest of the paper, assume
   that strategies are only defined for their accessible
   histories. Formally, we assume that a strategy is a partial
   function whose domain $H$ satisfies $\epsilon\in H$ and for all
   $h\in H$ and all $o\in \OOO$, $h.\lambda(h).o\in H$,
   and $H$ is minimal (for inclusion) w.r.t. this property.}.

A \emph{universal co-B\"uchi tree automaton} (UCT) is a tuple $\TTT = (\Sigma, Q,Q_0, D, \Delta, \alpha )$ where 
$\Sigma$ is the finite alphabet,
$Q$ is a finite set of states, 
 $Q_0 \subseteq Q$ is the set of initial states,
 $D$ is the set of directions,
 $\Delta : Q \times \Sigma \times D \rightarrow 2^{Q}$ is the transition relation
(assumed to be total) and $\alpha$ is the set of final states. If the
tree automaton is in some state $q$ at some node $x$ labelled
by some $\sigma\in\Sigma$, it will evaluate, for all $d\in D$,  the subtree
rooted at $x.d$ in parallel from all the states of
$\Delta(q,\sigma,d)$. Let us define the notion of run formally. 
For all $q\in Q$ and $\sigma\in \Sigma$, we
denote by $\Delta(q,\sigma) = \{ (q_1,d_1),\dots,(q_n,d_n)\}$ the \emph{disjoint} union of all
sets $\Delta(q,\sigma,d)$ for all $d\in D$. A \textit{run} of $\TTT$ on an infinite $\Sigma{-}labelled$ $D{-}$tree $\langle T, \tau \rangle$ is a $(Q \times D^*){-}labelled$ $\mathbb{N}{-} tree$ $ \langle T_r,\tau_r \rangle$ such that $\tau_r(\epsilon) \in Q_0 \times \{\epsilon \}$
and, for all $x \in T_r$ such that $\tau_r(x)=(q,v)$, if
$\Delta(q,\tau(v)) = \{(q_1,d_1),\dots,(q_n,d_n)\}$,
we have $x \cdot i \in T_r$ and $\tau_r(x \cdot i) = (q_i, v
\cdot d_i)$ for all $0 < i \leq n$, .
Note that there is at most one run per input tree (up to tree isomorphism). A run $\langle
T_r,\tau_r\rangle$ is \emph{accepting} if for all infinite branches $\pi$ of
$T_r$, $\tau_r(\pi)$ visits a finite number of accepting states. The
language of $\TTT$, denoted by $\lang_{uc}(\TTT)$,  is the set of $\Sigma-$labelled $D$-trees such that
there exists an accepting run on them. Similarly, we define
universal $B$-co-B\"uchi tree automata by strengthening the acceptance
conditions on all branches to the $B$-co-B\"uchi condition.

As noted in \cite{Filiot11,ScheweF07a}, testing the emptiness of a UCT automaton reduces
to testing the emptiness of a universal $B$-co-B\"uchi accepting
condition for a sufficiently large bound $B$, which in turn reduces to
solving a safety game. Symbolic techniques that are also exploited in
this paper have been used to solve the safety games.

\vspace{-2mm}
\section{LTL synthesis under imperfect information}\label{sec:ltl}
\vspace{-2mm}

In this section, we first explain an automata-based procedure to
decide realizability under imperfect information of LTL formulas 
against an environment model. This procedure will be extended in the next
section to handle the $K$ operator.

Take an environment model $\modelEnv =
(\Prop, \Sigma_1,\Sigma_2,S_e, S_0,\Delta_e,\tau_e)$. Then, a complete $\Sigma_1-$labelled $\OOO$-tree $\langle T,\tau\rangle$ defines a
strategy of the system. Any infinite branch $\pi$ of $\langle
T,\tau\rangle$ defines an infinite sequence of actions and
observations of $\modelEnv$, which in turn corresponds to a set of possible traces in 
$\modelEnv$. We denote by $\traces(\pi)$ this set of traces, and it is
formally defined by $\traces(\pi) = \traces(\tau(\pi))$ (recall that
the set of traces of a sequence of actions and observations has been
defined in Section \ref{sec:KLTLReal}).

Given an LTL formula $\psi$, we 
construct a universal co-B\"uchi tree automaton $\TTT = (\Sigma_1, Q, Q_0, \OOO,
\Delta, \alpha )$ that accepts all the strategies of Player $1$ (the system) that
realize $\psi$ under the environment model $\modelEnv$. First, one
converts $\psi$ into an equivalent UCW $\aut = (2^\Prop, Q^\aut,
Q_0^\aut, \Delta^\aut, \alpha^\aut)$.
Then, as a direct consequence of the
definition of KLTL realizability:

\begin{proposition}
    Given a complete $\Sigma_1-$labelled $\OOO$-trees $\langle
    T,\tau\rangle$, $\langle
    T,\tau\rangle$ defines a strategy that realizes $\psi$ under
    $\modelEnv$ iff for all infinite branches $\pi$ of $\langle
    T,\tau\rangle$ and all traces $\rho\in\traces(\pi)$, $\rho\in L(\aut)$.
\end{proposition}

We now show how to construct a universal tree automaton that checks
the property mentioned in the previous proposition, for all branches of
the trees. We use universal transitions to check, on every
branch of the tree, that all the possible traces (possibly uncountably
many) compatible with the sequence of actions in $\Sigma_1$ and
observations in $\OOO$ defined by the branch satisfy $\psi$. 
Based on finite sequences of 
observations that the system has received, it can define its
\emph{knowledge} $I$ of the possible states in which the environment can
be as a subset of states of $S_e$. 
Given an action $a_1\in\Sigma_1$ of the system and some observation
$o\in\OOO$, we denote by $\post_a(I,o)$ the new knowledge that the
system can infer from the observation $o$, the action $a$ and its
previous information $I$. Formally, $\post_{a}(I,o) = \{ s\in S_e \cap o \ |\
\exists a_2\in\Sigma_2, \exists s'\in I \text{ s.t. }(s',a,a_2,s)\in
\Delta_e\}$. 

The states of the universal tree automaton $\TTT$ are pairs of states
of $\aut$ and knowledges plus some extra state $(q_w, \varnothing)$, i.e. $Q = Q^\aut\times
2^{S_e} \cup \{ (q_w, \varnothing)\}$ where $(q_w, \varnothing)$ is added for completeness. The final
states are defined as $\alpha = \alpha^\aut\times 2^{S_e}$ and initial
states as $Q_0 = Q_0^\aut\times S_0$. 
To define the transition relation, let us consider a state 
$q\in Q^\aut$, a knowledge set $I\subseteq S_e$, an action
$a\in\Sigma_1$ and some observation $o\in\OOO$. 
We now define $\Delta((q,I), a, o)$. It could be the case that there is no transition in $\modelEnv$ 
from a state of $I$ to a state of $o$, i.e. $\post_{a}(I,o) =
\varnothing$. In that case, all the paths from the next $o$-node 
of the tree should be accepting. This situation is modelled by going to the extra
state $(q_w,\varnothing)$, i.e. $\Delta((q,I),a,o) =
(q_w,\varnothing)$.

Now suppose that $\post_{a}(I,o)$ is non-empty. Since the automaton must
check that all the traces of $\modelEnv$ that are compatible with
actions of $\Sigma_1$ and observations are accepted by $\aut$,
intuitively, one would define $\Delta((q,I),a,o)$ as the set of
states of the form $(q', \post_{a}(I,o))$ for all states $q'$ such that
there exists $s\in I$ such that $(q,\tau_e(s),q')\in
\Delta^\aut$. However, it is not correct for several reasons. First,
it could be that $s$ has no successor in $o$ for action $a$, and therefore one
should not consider it because the traces up to state $s$ die at the
next step after getting observation $o$. Therefore, one should only
consider states of $I$ that have a successor in $o$. Second, it is not
correct to associate the new knowledge $\post_{a}(I,o)$ with $q'$,
because it could be that there exists a state $s'\in \post_{a}(I,o)$
such that for all its predecessors $s''$ in $I$, there is not
transition $(q,\tau_e(s''), q')$ in $\Delta^\aut$, and
therefore, one would also take into account sequences of
interpretations  of propositions that do not correspond to
any trace of $\modelEnv$.

Taking into account these two remarks, we
define, for all states $q'$, the set $I_{q,q'} = \{ s\in I\ |\
(q,\tau(s),q')\in \Delta^\aut\}$. Then, $\Delta((q,I),a,o)$ is
defined as the set 

$$\Delta((q,I),a,o) = \{ (q', \post_{a}(I_{q,q'}, o))\ |\ \exists s\in
I,\ (q,\tau(s),q')\in \Delta^\aut\}$$

\noindent Note that, since
$\bigcup_{q'\in Q^\aut} \post_a(I_{q,q'}, o) = \post_a(I, o)$
and the
automaton is universal, the system does not have better
knowledge by restricting the knowledge sets. 

\begin{lemma}\label{lemma:LTLReal}
    The LTL formula $\psi$ is realizable in $\modelEnv$ iff
    $\lang(\TTT)\neq \varnothing$. 
\end{lemma}

Moreover, it is known that if a UCT has a non-empty language, then
it accepts a tree that is the unfolding of a finite graph, or
equivalently, that can be represented by a Moore machine. Therefore 
if $\psi$ is realizable, it is realizable by a finite-memory
strategy.

\vspace{-4mm}
\section{Safraless procedure for positive KLTL synthesis}\label{sec:kltl}
\vspace{-2mm}

In this section, we extend the construction of Section \ref{sec:ltl}
to the positive fragment of KLTL. Positive formulas are 
defined by the following grammar:

\vspace*{-2mm}
\begin{center} 
$\varphi ::= p \mid \neg p \mid \varphi \wedge \varphi \mid \varphi \vee \varphi \mid \bigcirc \varphi \mid \Box \varphi \mid  K \varphi \mid \varphi \UUU \varphi$\vspace*{-0.5mm}
\end{center}
Note that this fragment is equivalent to the fragment of \KLTL in which 
formulas with the knowledge operator $K$ occurring under an even number of negations. 
This is obtained by straightforwardly pushing the negations down towards the atoms.
 We denoted this fragment of KLTL by \pKLTL.

%

\vspace*{-2mm}
\paragraph{\textbf{Sketch of the construction}} Given a \pKLTL\ formula
$\varphi$ and an environment model $\modelEnv =
(\Prop,\Sigma_1,\Sigma_2,S_e, S_0,\Delta_e,\tau_e)$, we show how to construct a UCT $\TTT$ such that
$\lang(\TTT)\neq \varnothing$ iff $\varphi$ is realizable in
$\modelEnv$. The construction is compositional and follows, for the
basic blocks, the construction of Section \ref{sec:ltl} for LTL
formulas. The main idea is to replace subformulas of the form 
$K\gamma$ by fresh atomic propositions $k_\gamma$ so that we get an LTL
formula for which the realizability problem can be transformed into
the emptiness of a UCT. The
realizability of the subformulas $K\gamma$ that have been replaced by
$k_\gamma$ is checked by branching universally to a UCT for $\gamma$,
 constructed as in Section \ref{sec:ltl}. Since transitions are
universal, this will ensure that all the infinite branches of the
tree from the current node where a new UCT has been triggered also
satisfy $\gamma$. The UCTs we construct are
defined over an extended alphabet that contains the new atomic
propositions, but we show that we can safely project the final UCT on the
alphabet $\Sigma_1$. The assumption on positivity
of KLTL formulas implies that there is no subformulas of the form
$\neg K\gamma$. The rewriting of subformulas by fresh atomic
propositions cannot be done in any order. We now describe it formally.

 %
%
%


We inductively define a sequence of formulas associated with
$\varphi$ as: $\varphi^0 = \varphi$ and, for all $i>0$, $\varphi^i$ is
the formula $\varphi^{i-1}$ in which the \textit{innermost}
subformulas $K \gamma$ are replaced by fresh atomic propositions
$k_\gamma$. Let $d$ be the smallest index such that $\varphi^d$ is an
LTL formula (in other words, $d$ is the maximal nesting level of $K$
operators). Let $\KKK$ denote the set of new atomic propositions, i.e., 
$\KKK = \bigcup_{i = 0}^d \{ k_\gamma \mid K \gamma \in \varphi^i \}$,
and let $\Prop' = \Prop \cup \KKK$. Note that by definition of the
formulas $\varphi^i$, for all atomic proposition $k_\gamma$ occurring
in $\varphi^i$, $\gamma$ is an LTL formula over $\Prop'$. 
E.g. if $\varphi = p \rightarrow K(q \rightarrow K r \vee Kz)$ and $\Prop = \{ p,q,r,z \}$, then the sequence of formulas is: 
$\varphi^0 = \varphi$, 
$\varphi^1 = p \rightarrow K(q \rightarrow k_{r} \vee k_{z})$ 
$\varphi^2 = p \rightarrow k_{\gamma} \text{ where } \gamma = q \rightarrow k_{r} \vee k_{z}$. 

Then, we construct incrementally a chain of universal co-B\"uchi tree automata
$\TTT^d,\dots,\TTT^0$ such that $\lang(\TTT^d)\supseteq
\lang(\TTT^{d{-}1})\supseteq\dots\supseteq \lang(\TTT^0)$ and, 
the following invariant is satisfied: for all $i\in\{0,\dots,d\}$, $\TTT^i$ accepts exactly the set of
strategies that realize $\varphi^i$ in $\modelEnv$. Intuitively, the
automaton $\TTT^i$ is defined by adding new transitions in
$\TTT^{i+1}$, such that for all atomic propositions $k_\gamma$
occurring in $\varphi^{i+1}$, $\TTT^i$ will ensure that $K\gamma$ is
indeed satisfied, by branching to a UCT checking $\gamma$ whenever the
atomic proposition $k_\gamma$ is met. Since formulas 
$\varphi^i$ are defined over the extended alphabet $\Prop' = \Prop\cup
\KKK$ and $\modelEnv$ is defined over $\Prop$, we now make clear what
we mean by realizability of a formula $\varphi^i$ in $\modelEnv$. It
uses the notion of \emph{extended model executions} and \emph{extended
  strategies}. 

\vspace*{-2mm}
\paragraph{\textbf{Extended actions, model executions and strategies}} We extend the actions of the
system to $\Sigma'_1 = \Sigma_1 \times 2^\KKK$ (call $e$-actions). Informally, the system
plays an $e$-action $(a,K)$ if it considers formulas $K\gamma$ for all
$k_\gamma\in K$ to be true. An \emph{extended execution}
($e$-execution) of $\modelEnv$ is a infinite sequence $\rho = (s_0,K_0)... \in
(S_e\times 2^\KKK)^\omega$ such that $s_0s_1... \in
\exec(\modelEnv)$. We denote $s_0s_1\dots $ by $\proj_1(\rho)$ and
$K_0K_1\dots$ by $\proj_2(\rho)$. 
The extended labelling function $\tau_e' $
is a function from $S_e\times 2^\KKK$ to $\Prop'$ defined by
$\tau_e'(s,K) = \tau_e(s)\cup K$. The indistinguishability relation
between extended executions is defined, for any two extended
executions $\rho_1,\rho_2$, by $\rho_1\sim \rho_2$ iff
$\proj_1(\rho_1)\sim \proj_1(\rho_2)$ and $\proj_2(\rho_1) = \proj_2(\rho_2)$, i.e., the propositions in
$\KKK$ are visible to the system. We define $\sim_i$ over extended
executions similarly. Given the extended labelling functions and 
indistinguishability relation, the KLTL satisfiability notion
$R,\rho,i\models \psi$ can be naturally defined for a set of
$e$-execution $R$, $\rho\in R$ and $\psi$ a KLTL formula over 
$\Prop' = \Prop\cup \KKK$.

An extended strategy is a strategy defined over $e$-actions, i.e. a function from $(\Sigma'_1\OOO)^*$ to
$\Sigma'_1$. For an infinite sequence $u =
(a_0,K_0)o_0(a_1,K_1)o_1\dots\in (\Sigma'_1\OOO)^\omega$, we define 
$\proj_1(u)$ as $a_0o_0\dots$.  The sequence $u$ defines a set
of compatible $e$-executions $\exec(\modelEnv, u)$ as follows: it is
the set of $e$-executions $\rho = (s_0,K_0)(s_1,K_1)... \in (S_e\times
2^\KKK)^\omega$ such that $\proj_1(\rho)\in \exec(\modelEnv, \proj_1(u))$. 
Similarly, we define for $e$-strategies $\lambda'$ the set
$\exec(\modelEnv, \lambda')$ of $e$-executions compatible with
$\lambda'$. A KLTL formula $\psi$ over $\Prop'$ is realizable in
$\modelEnv$ if there exists an $e$-strategy $\lambda'$ such that for
all runs $\rho\in \exec(\modelEnv, \lambda')$, we have
$\exec(\modelEnv, \lambda'), \rho, 0\models \psi$.

\begin{proposition}\label{prop:basecase}
    There exists an $e$-strategy $\lambda':(\Sigma_1'\OOO)^*\rightarrow
    \Sigma_1'$ realizing $\varphi^0$ in $\modelEnv$ iff there exists a strategy
    $\lambda:(\Sigma_1\OOO)^*\rightarrow \Sigma_1$ realizing 
    $\varphi^0$ in $\modelEnv$.
\end{proposition}

\begin{proof}
    Let see $e$-strategies and strategies as $\Sigma'_1-$labelled
    (resp. $\Sigma_1$-labelled) $\OOO$-trees. Given a tree
    representing
    $\lambda'$, we project its labels on $\Sigma_1$ to get
    a tree representing $\lambda$. The strategy $\lambda$ defined in
    this way realises $\varphi^0$, as $\varphi^0$ does not contain
    any occurrence of propositions in $\KKK$. Conversely, given a tree
    representing $\lambda$, we extend its labels with $\emptyset$ to
    get a tree representing $\lambda'$. It can be shown for the
    same reasons that $\lambda'$ realizes $\varphi^0$. \qed


\end{proof}

\vspace*{-2mm}
\paragraph{\textbf{Incremental tree automata construction}}
The invariant mentioned before can now be stated more precisely: for all $i$, $\TTT^i$ accepts the
$e$-strategies $\lambda':(\Sigma_1'\OOO)\rightarrow \Sigma_1'$ that
realise $\varphi^i$ in $\modelEnv$. Therefore, the UCT $\TTT^i$ are
labelled with $e$-actions $\Sigma'_1$. We now explain how they are
constructed.

Since $\varphi^d$ is an LTL formula, we follow the construction of
Section \ref{sec:ltl} to build the UCT $\TTT^d$. Then, we construct $\TTT^i$ from
$\TTT^{i+1}$, for $0\leq i<d$. The invariant tells us that $\TTT^{i+1}$ defines all the
$e$-strategies that realize $\varphi^{i+1}$ in $\modelEnv$. It is only
an over-approximation of the set of $e$-strategies that realize $\varphi^{i}$ in
$\modelEnv$ (and \emph{a fortiori} $\varphi^0$), since the subformulas of $\varphi^{i}$ of the form
$K\gamma$ correspond to atomic propositions $k_\gamma$ in
$\varphi^{i+1}$, and therefore $\TTT^{i+1}$ does not check that they
are satisfied. Therefore to maintain the invariant, $\TTT^i$ is
obtained from $\TTT^{i+1}$ such that whenever an action that contains
some formula $k_\gamma\in sub(\varphi^{i+1})$ occurs
on a transition of $\TTT^{i+1}$, we trigger (universally) a new
transition to a UCT $\TTT_{\gamma,I}$, for the current information set $I$ in $\TTT^{i+1}$,
 that will check that $K\gamma$ indeed holds. 
The assumption on positivity of KLTL formulas is necessary here as we
do not have to check for formulas of the form $\neg K\gamma$, which
could not be done without an involved ``non Safraless'' complementation step.  
Since $\gamma$ is necessarily an LTL formula over $\Prop'$ by definition of
the formula $\varphi^{i+1}$, we can apply the construction of
Section \ref{sec:ltl} to build $\TTT_{\gamma,I}$.

Formally, from the incremental way of constructing the automata
$\TTT^j$ for $j\geq i$,  we know that $\TTT^{i+1}$ has a set of states $Q_{i+1}$ where all
states are of the form $(q,I)$ where $I\subseteq S_e$ is some
knowledge. In particular, it can be verified to be true for the state
space of $\TTT^d$ by definition of  the construction of Section
\ref{sec:ltl}. Let also $\Delta_{i+1}$ be the transition relation of
$\TTT^{i+1}$. For all formulas $\gamma$ such that $k_\gamma$ occurs in
$\varphi^{i+1}$, we let $Q_\gamma$ be the set of states of
$\TTT_{\gamma,I}$ and $\Delta_\gamma$ its set of transitions. 
Again from the construction of Section \ref{sec:ltl}, we know that 
$Q_\gamma = Q\times 2^{S_e}$ where $Q$ is the set of states of a UCW associated with
$\gamma$ (assumed to be disjoint from that of $\TTT^{i+1}$) and $Q_\gamma^0 = Q \times I$.

We define the set of states $Q_i$ of $\TTT^i$ by $Q_{i+1}\cup
\bigcup_{k_\gamma\in sub(\varphi^{i+1})} Q_\gamma$. Its set of
transitions $\Delta_i$ is defined as follows. Assume
w.l.o.g. that there is a unique initial state $q_0\in
Q$ in the UCW $\aut_\gamma$. If $(q',I')\in \Delta_{i+1}((q,I),(a,K),o)$ where $I,I'\subseteq S_e$,
$a\in\Sigma_1$, $K\subseteq \KKK$, $o\in \OOO$ and $k_\gamma\in K$ is such that $k\gamma$ occurs in
$\varphi^{i+1}$, then we let $(q',I')\in \Delta_i((q,I),(a,K),o)$ and
$\Delta_{\gamma}((q_0,I),(a,K),o)\subseteq
\Delta_i((q,I),(a,K),o)$. The whole construction is given in Appendix,
as well as the proof of its correctness. The invariant is satisfied:

\begin{lemma}\label{lem:intermediatereal}
    For all $i\geq 0$, $\lang(\TTT^i)$ accepts the set of $e$-strategies
    that realize $\varphi^i$ in $\modelEnv$. 
\end{lemma}

    From Lemma \ref{lem:intermediatereal}, we know that
    $\lang(\TTT^0)$ accepts the set of $e$-strategies that realize
 $\varphi^0 = \varphi$ in $\modelEnv$. Then by Proposition
 \ref{prop:basecase} we get:

\begin{corollary}\label{coro:realiz}
    The $\pKLTL\ $ formula $\varphi$ is realizable in $\modelEnv$
    iff $\lang(\TTT^0)\neq \varnothing$. 
\end{corollary}

We now let $\TTT_\varphi$ be the UCT obtained by projecting
$\TTT^0$ on $\Sigma_1$. We have:

\begin{theorem}\label{thm:main}
    For any $\pKLTL$ formula $\varphi$, one can construct a UCT 
    $\TTT_\varphi$ such that $\lang(\TTT_\varphi)$ is the set of strategies that realize $\varphi$
    in $\modelEnv$.
\end{theorem}

The number of states of $\TTT_\varphi$ is 
(in the worst-case) $2^{|S_e|}.(2^{|\varphi^d|} + \sum_{k_\gamma \in
  \KKK} 2^{|\gamma|})$, and since $|\varphi^d| + \sum_{\gamma \in
  \KKK} |\gamma|$ is bounded by $|\varphi|$, the number of states of 
$\TTT_\varphi$ is $O(2^{|S_e| + |\varphi|})$. 


\vspace{-2mm}
\section{Antichain Algorithm}\label{sec:antichains}
\vspace{-2mm}

In the previous sections, we have shown how to reduce the problem of 
checking the realizability of some \pKLTL\ formula $\varphi$ to the
emptiness of a UCT $\TTT_\varphi$ (Theorem \ref{thm:main}). In this
section, we describe an antichain symbolic algorithm to test the
emptiness of $\TTT_\varphi$.

It is already known from \cite{Filiot11}
that checking emptiness of the language defined by a UCT $\TTT$ can be
reduced to checking the emptiness of $\lang_{uc,B}(\TTT)$ for a sufficiently
large bound $B$, which in turn can be reduced to solving a safety
game. Clearly, for all $B\geq 0$, if $\lang_{uc,B}(\TTT)\neq
\emptyset$, then $\lang_{uc}(\TTT)\neq \emptyset$. This has led to an
incremental algorithm by starting with some small bound $b$ and the
experiments have shown that in general, a small bound $b$ is necessary
to conclude for realizability of an LTL formula (transformed into the
emptiness of a UCT). We also exploit this idea in our implementation
and show that for \pKLTL\ specifications that we considered, this
observation still holds: small bounds are enough.

In \cite{Filiot11}, it is shown that the safety games can be solve
on-the-fly without constructing them explicitly, and that the fixpoint
algorithm used to solve these safety games could be optimized by using
some antichain representation of the sets constructed during the
fixpoint computation. Rather than using the algorithm of
\cite{Filiot11} as a black box, we study the state space of the safety
games constructed from the UCT $\TTT_\varphi$ and show that they are
also equipped with a partial order that allows one to get more compact
antichain representations. We briefly recall the reduction of
\cite{Filiot11}, the full construction of the safety games is given in
Appendix.

Given a bound $b\geq 0$ and a UCT $\TTT$, the idea is to construct a
safety game $G(\TTT,b)$ such that Player $1$ has a winning strategy in
$G(\TTT,b)$ iff $L_{uc,b}(\TTT)$ is non-empty. The game $G(\TTT,b)$ is
obtained by extending the classical automata subset construction with
counters which count, up to $b$, the maximal number of times all the
runs, up to the current point, have visited accepting states. If $Q$
is the set of states of $\TTT$, the set of states of the safety game 
$G(\TTT,b)$ is all the functions $F : Q\rightarrow
\{-1,0,\dots,b+1\}$. 
The value $F(q) = -1$ means that no run have
reached $q$ and $F(q)\in\{0,\dots,b\}$ means that the maximal number
of accepting states that has been visited by the runs reaching $q$ is $F(q)$.
The safe states are all the functions $F$ such that $F(q)\leq b$
for all $q\in Q$. The set of states can be partially ordered by the
pairwise comparison between functions and it is shown that the 
sets of states manipulated by the fixpoint algorithm are downward
closed for this order.

Consider now the UCT $\TTT_\varphi$ constructed from the \pKLTL\ 
formula $\varphi$. Its state space is of the form $Q\times 2^{S_e}$
where $S_e$ is the set of states of the environment, because the
construction also take into account the knowledge the system has from
the environment. Given a bound $b$, the state space of the safety
game $G(\TTT_\varphi, b)$ is therefore functions from $Q\times
2^{S_e}$ to $\{-1,\dots,b+1\}$. However, we can reduce
this state space thanks to the following result:

\begin{proposition}\label{prop:knowledgebis}
For all runs $\langle T_r, \tau \rangle$ of $\TTT_\varphi$ on some tree
$T$, for all branches $\pi,\pi'$ in $T_r$ of the same length such that
they follow the same sequence of observations,  if
$\tau(\pi) = (q,I)$ and $\tau(\pi') = (q,I')$, then $I = I'$.
\end{proposition}

In other words, given the same sequence of observations, the tree
automaton $\TTT_\varphi$ computes, for a given state $q$, the same
knowledge.

Based on this proposition, it is clear that reachable states $F$ of
$G(\TTT_\varphi, b)$ satisfy, for all states $q\in Q$ and knowledges
$I,I'$, if $F(q,I)\neq -1$ and $F(q,I')\neq -1$ then $I = I'$. We can
therefore define the state space of $G(\TTT_\varphi, b)$ as the set of
pairs $(F, \overline{K})$ such that $F : Q\rightarrow
\{-1,\dots,b+1\}$ and $\overline{K} : Q\rightarrow 2^{S_e}$ associates
with each state $q$ a knowledge (we let $G(q) = \emptyset$ if $F(q) =
-1$).  This state space is naturally ordered by $(F_1,\overline{K}_1) \preceq
(F_2,\overline{K}_2)$ if for all $q\in Q$, $F_1(q)\leq F_2(q)$ and 
$\overline{K}_1(q) \subseteq \overline{K}_2(q)$. We show that all the
sets manipulated during the fixpoint computation used to solve the
safety games are downward closed for this order and therefore can be
represented by the antichain of their maximal elements. A detailed
analysis of the size of the safety game shows that $G(\TTT_\phi, B)$
is doubly exponential in the size of $\varphi$, and therefore, since
safety games can be solved in linear time, one gets a 2Exptime upper
bound for $\pKLTL\ $ realizability. The technical 
details are given in Appendix.



\vspace{-2mm}
\section{Implementation and Case Studies}\label{sec:implementation}
\vspace{-2mm}

In this section we briefly present our prototype implementation
\Acacia-K for \pKLTL synthesis \cite{AcaciaK}, and provide some interesting examples
on which we tested the tool, on a laptop equipped with an Intel Core
i7 2.10Ghz CPU. \Acacia-K extends the LTL synthesis tool
Acacia+\cite{CAV12}. As Acacia+, the implementation is made
in Python together with C for the low level operations that need efficiency.

As Acacia+, the tool is available in one version working on both Linux
and MacOsX and can be executed using the command-line interface.
As parameters, in addition to the files containing the $\pKLTL$
formula and the partition of the signals and actions, \Acacia-K
requires a file with the environment model. The output of the tool is
a winning strategy, if the formula is realizable, given as a Moore
machine described in Verilog
and if this strategy is small,  \Acacia-K also outputs it as a
picture.

In order to have a more efficient implementation, the construction of the automata for the LTL formulas $\gamma$ 
is made on demand. That is, we construct the UCT $\TTT_\gamma$ incrementally
by updating it as soon as it needs to be triggered from some state
$(q,I)$ which has not been constructed yet.

As said before, the synthesis problem is reduced to the
problem of solving a safety game for some bound $b$ on the number of
visits to accepting states. The tool is incremental: it tests realizability for small
values of $b$ first and increments it as long as it cannot conclude for
realizability. In practice, we have observed, as for classical LTL
synthesis, that small bounds $b$ are sufficient to conclude for
realizability. However if the formula is not realizable, we have to
iterate up to a large upper bound, which in practice is too large to
give an efficient procedure for testing unrealizability. We leave as
future work the implementation of an efficient procedure for testing
unrealizability.


Taking now Example \ref{ex:light}, the strategy provided by the tool
is depicted in Figure \ref{fig:str-light}. 
It asks to play first "toggle" and then keep play "skip" and,
depending on the observation he gets, the system goes in a different
state. The state $0$ is for the start, the state $1$ is the "error"
state in which the system goes if he receives a wrong observation. 
That is, the environment gives an observation even if he cannot go in
a state having that observation. 
Then, if the observation is correct, after playing the action "toggle"
from the initial states $\{ s_1, s_2 \}$, the environment is forced to
go in $s_3$ and by playing the action "skip", the system
forces the environment to stay in $s_3$ and he will know that $t$ is
false. In the strategy, this situation corresponds to the state $2$.
For this example, \Acacia-K constructed a UCT with 31 states 
and the total running time is 0.2s.

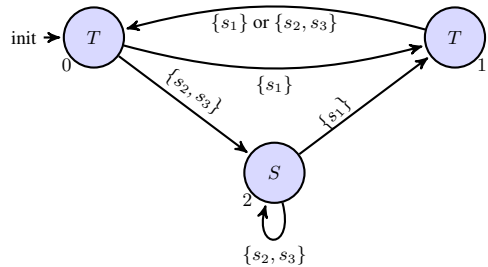
\begin{figure}
\centering
\vspace*{-10pt}
\begin{center}
\begin{tikzpicture}[->,>=stealth',shorten >=1pt,auto,node distance=2.5cm,
                    thick,scale=0.6,every node/.style={scale=0.8}]
  \tikzstyle{every
    state}=[shape=circle,fill=blue!15,text=black,minimum size=1cm]
  \tikzstyle{every edge}=[draw=black]
  \tikzstyle{initial}=[initial by arrow, initial where=left, initial text=init]

  \node[initial,state] (s1) at (0,0) {$T$};
  \node[state] (s3)  at (4,-3) {$S$};
  \node[state] (s2)  at (8,0) {$T$};

  \node[state,draw=none,fill=none] (s1l) at (-0.6,-0.6) {$0$};
  \node[state,draw=none,fill=none] (s3l)  at (3.4,-3.6) {$2$};
  \node[state,draw=none,fill=none] (s2l)  at (8.6,-0.6) {$1$};

\path (s3) [loop below] edge node {$\{s_2, s_3 \}$} (s3) ;

\path (s1) [bend right=15] edge node[below]{$\{s_1 \}$} (s2) ;
\path (s2) [bend right=15] edge node {$\{ s_1 \}$ or  $\{ s_2, s_3 \}$} (s1) ;

\path (s1)  edge node[rotate=-35, above] {$\{ s_2, s_3 \}$} (s3) ;

\path (s3)  edge node[rotate=48] {$\{s_1\}$} (s2) ;

\vspace{-10mm}
\end{tikzpicture}
\end{center}
\vspace*{-10pt}
\caption{Winning strategy synthesized by \Acacia-K for Example \ref{ex:light}}
\label{fig:str-light}
\vspace*{-10pt}
\end{figure}



\begin{example}[The 3-Coin Game]
Another example that we tried is a game played using three coins which
are arranged on a table with either \textit{head} or \textit{tail}
up. The system doesn't see the coins, but knows at each time the
number of tails and heads. Then, the game is infinitely played as follows. 
At the beginning the environment chooses
an initial configuration and then at each round, the system chooses a
coin and the environment has to flip that coin
and inform the system about the new number of heads and tails. 
The objective of the system is to reach, at least once, the state in
which all the coins have the \textit{heads} up and to avoid all the
time the state in which all the coins are \textit{tails}. Depending on
the initial number of tails up, the system may or may not have a
winning strategy.
\end{example}

In order to model this, we considered an environment model whose
states are labelled with atomic propositions $c_1, c_2, c_3$ for the three coins, which are not visible for the system, and two other variables $b_1, b_0$ which are visible and represent the bits encoding the number of \textit{heads} in the configuration. 
The actions of the system are $C_1, C_2, C_3$ with which he chooses a coin and  
the environment has to flip the coin chosen by the system by playing only the action $done$.
A picture of the environment is in Figure \ref{fig:env-coins} from Appendix .

Then, the specification is translated into the \pKLTL formula
$K \Diamond (c_1 \wedge c_2 \wedge c_3 ) \wedge \Box K( c_1 \vee c_2 \vee c_3)$.
Then, assuming that the initial state of the environment has two \textit{heads}, 
the synthesized strategy proposes to "check" the position of every coin by double flipping. 
If after one flip, the winning state is not reached, the system flips back the coin 
and at the third round he chooses another coin to check.
A picture of the strategy can be found in Figure \ref{fig:str-coins} of Appendix.
For this example, \Acacia-K constructs a UCT with 79 states,
synthesises a strategy with 10 states, and the total running time is
3.9s.

Finally, we have designed an example (the \emph{prisoners enigma}). It
is not presented in the paper but can be found in Appendix. 
We have tried 3/4/5/6 prisoners versions (including the protagonist) of this problem, obtaining a one hour timeout for 6 agents. 
The statistics we obtained are the following:
\begin{center}
\begin{tabular}{ | l | l | l | l | l | l | l |}
	
  Pris \# & $|\modelEnv|$ & $|UCT|$ & $|tb-UCT|$ & Aut constr (s) & $|\mathcal{M}_\lambda|$ & Total time(s)\\
  \hline                       
   3 & 21 & 144 & 692 & 1.79s & 12 & 1.87s\\
   4 & 53 & 447 &  2203 & 1.98s & 16 & 13.20s  \\
   5 & 129 & 1310 & 6514 & 199.06s & 20 & 553.45s ($\simeq$ 9 min)\\
   6 & 305 & 3633 & 18125 & 6081.69s & N/A & N/A \\
  \hline  
\end{tabular}
\end{center}
Again, Acacia-K generates strategies that are natural, the same that one would synthesize intuitively. This fact is remarkable itself since, in synthesis, it is often a difficult task to generate small and natural strategies.





\vspace{-4mm}
\section{Conclusion}
\vspace{-2mm}

In this paper, we have defined a Safraless procedure for the synthesis of \pKLTL
specifications in environment with imperfect information. This problem
is \textsf{2ExpTime-c} but we have shown that our procedure, based on
universal co-B\"uchi tree automata, can be implemented efficiently
thanks to an antichain symbolic approach. We have implemented a
prototype and run some preliminary experiments that prove 
the feasibility of our method. While the
UCT constructed by the tool are not small (around 1300 states), our
tool can handle them, although in theory, the safety games could be
exponentially larger than the UCT. Moreover, our tool synthesises small
strategies that correspond to the intuitive strategies we would
expect, although it goes through a non-trivial automata construction.
As a future work, we want to see if \Acacia-K scales well
on larger examples. We also want to extend the tool to handle the full
KLTL logic in an efficient way. This paper is an encouraging (and
necessary) step towards this objective. In a first attempt to
generalize the specifications, we plan to consider assume-guarantees specifications
$K\phi\rightarrow \psi$, where $\phi$ is an LTL formula and $\psi$ a
$\pKLTL$ formula.

\vspace*{-5mm}
\bibliographystyle{abbrv}
\bibliography{lics-bibtex-full}
\nocite{HalpernVardi86,lomuscio-mcmas,kacprzak-penczek-aamas-2005,bulling-jamroga-mu,goranko-drimmelen06,agotnes-synthese06,arnold-rudiments,walukiewicz-tcs-mu-calculus}

\newpage
\appendix
\section{Correctness of the UCT construction for
 LTL formulas(proof of Lemma \ref{lemma:LTLReal})}
\vspace*{-2mm}

\begin{proof}
 If $\psi$ is realizable in $\modelEnv$, there exists a strategy $\lambda: (\Sigma_1 \times \OOO)^* \rightarrow \Sigma_1$ 
for the system such that $\exec(\modelEnv, \lambda) \models \psi$.
 Let's see this strategy as a $\Sigma_1$-labelled $\OOO$-tree $\langle T_\lambda, \tau \rangle$ and prove that $T_\lambda \in \lang(\TTT)$.

The run on $T_\lambda$ of the automaton $\TTT$ is a $Q \times \OOO^*$-labelled $\mathbb{N}$-tree 
$\langle T_r, \tau_r \rangle$.
Therefore, each branch $\pi$ of $T_r$ induces an infinite sequence 
$\tau_r(\pi) = ((q_0, I^{(0)}),\epsilon) ((q_1, I^{(1)}),o_1) ((q_2, I^{(2)}),o_1o_2)...$
where $(q_0,I^{(0)}) \in Q_0$ and $I^{(i+1)} = \post_{\tau(o_1\dots o_i)}
(I^{(i)}_{q_i, q_{i+1}}, o_{i+1})$, by definition of the transition
relation of $\TTT$.

Since $I^{(i)}_{q_i, q_{i+1}}$ is the subset of environment state in
$I^{(i)}$ whose labels can fire transitions from $q_i$ to $q_{i+1}$ in $\aut$,
$r = q_0 q_1 q_2 ...$ is a run in $\aut$ on the traces of the executions $\rho = s_0 s_1 s_2 ... $  
where $s_i \in I_i$.
Hence, from the fact that $\exec(\modelEnv, \lambda) \models \psi$ and $\rho \in \exec(\modelEnv, \lambda) $, $r$ is an accepting run 
and then $\tau_r(\pi)$ visits a finite number of accepting states in $\TTT$. 
Therefore, $T_\lambda \in \lang(\TTT)$ and $\lang(\TTT) \not = \emptyset$.

 In the other direction, if $\lang(\TTT) \not = \emptyset$, there
 exists a $\Sigma_1$-labelled $\OOO$-tree $\langle T_\lambda, \tau
 \rangle$ such that $T_\lambda \in \lang(\TTT)$. We have to prove that
 the strategy $\lambda$ defined by $T_\lambda$ realizes $\psi$, i.e.,
 for all executions $\rho\in \exec(\modelEnv, \lambda)$, we have
$\exec(\modelEnv, \lambda), \rho, 0\models \psi$. Let
$\rho\in\exec(\modelEnv, \lambda)$. Since $\psi$ is an
LTL formula, the sets of executions in which $\psi$ is evaluated does
not matter, and therefore we have to prove that $\rho, 0\models \psi$
(where $\models$ here denotes the classical LTL semantics). In other
words, we have to show that $\trace(\rho)\in \lang(\aut)$, i.e. all
the runs of $\aut$ on $\trace(\rho)$ visit finitely many accepting
states. Let $r = q_0q_1\dots$ be a run of $\aut$ on $\trace(\rho)$.

Let $\langle T_r, \tau_r \rangle$ be the (accepting) run of $\TTT$ on
$\langle T_\lambda, \tau\rangle$. It is a  $Q \times \OOO^*$-labelled
$\mathbb{N}$-tree. By definition of $\exec(\modelEnv, \lambda)$, there
exists a sequence $u = a_0o_0\dots \in \Sigma_1\times \OOO$ such that 
$\rho = s_0 s_1 s_2 ...$ is compatible with $u$. Let define $t$ as the infinite
sequence $t = (q_0, I^{(0)}),\epsilon) ((q_1, I^{(1)}),o_1) ((q_2,
I^{(2)}),o_1o_2)...$ such that $s_i \in I^{(i)}\subseteq S$ for all $i\geq
0$, $I^{(0)}=(q_0,S_0)$ and $I^{(i+1)} = \post_{a_i} (I^{(i)}_{q_i, q_{i+1}}, o_{i+1})$ for
all $i\geq 0$. 
It is easily shown that this sequence is such that
there exists a branch $\pi$ in  $\langle T_r, \tau_r \rangle$ with
$\tau_r(\pi) = t$ since the automata $\aut$ and $\TTT$ are complete 
and $\forall i, \bigcup_{q'\in Q} \post_a(I^{(i)}_{q_i,q'}, o) = \post_a(I^{(i)}, o)$.
Therefore there are finitely many
accepting states in $q_0q_1\dots$, since $T_r$ is accepting.
Therefore $\trace(\rho) \in \lang(\aut)$.



\qed
\end{proof}

\vspace*{-2mm}
\section{Correctness of the UCT construction for $\pKLTL$}
\vspace*{-2mm}

In the following, we prove that the automaton $\TTT^i$ accepts exactly the 
strategies that realize $\varphi^i$ in the extended runs of the environment $\modelEnv$.
A strategy of the system in this case can be seen as 
a $\Sigma_1 \times 2^{\KKK}$-labelled $\OOO$-tree $T$. 
A branch of $T$ is a sequence $\pi = \tilde{a}_1 o_1 \tilde{a}_2 o_2 ...$ where 
$\forall i$, $\tilde{a}_i \in \Sigma_1 \times 2^\KKK$ and $o_i \in \OOO$.

As we mentioned before, we start our construction with the LTL formula $\varphi^d$ 
and construct the universal co-B\"{u}chi tree automaton $\TTT^d = \TTT_{\varphi^d,S_0}$ 
for the LTL formula $\varphi^d$. 
Then, we construct the sequence of automata $\TTT^d$, $\TTT^{d-1}$,... $\TTT^0$ such that $\TTT^i$ checks for the satisfaction of $\varphi^i$ on the branches of accepted trees.
In other words, $\forall T \in \lang(\TTT^i)$,  $\forall \rho \in \R{T}$, we have that 
$$\R{T}, \rho,0 \models \varphi^i$$

In order to prove that the invariant holds, 
let define 
\[
\R{T, \pi, j} = \bigcup_{\pi' \in T : \pi'[0...j]=\pi[0...j]} \R{ \pi' }
\]
be the set of all extended executions that are compatible with the branch $\pi$ of the tree $T$ up to position $j$. 
This runs are indistinguishable from the runs in $\R{\pi}$ up to position $j$. 
Observe that by definition we have $\R{T,\pi,0} = \R{T} $.

\begin{definition}
A $\Sigma_1 \times 2^{\mathbb{K}}$-labelled $\OOO$-tree $T$ is called \textit{fair} with respect to a $K$-positive one-agent KLTL formula $\phi$ if:
\begin{align*}
&\forall \pi = (a_0,K_0) o_0 (a_1,K_1) o_1 ..., 
\forall j, \forall k_{\gamma} \in K_j \cap sub(\phi), \\
&\forall \pi' \text{ s.t. } \pi'[0...j] = \pi[0...j], \\
&\forall \rho \in \R{ \pi' }, \text{ we have } \R{ \pi' }, \rho, j \models \gamma
\end{align*} 
\end{definition}
\vspace*{-5pt}
Intuitively, whenever we have $k_{\gamma}$ in a node of a branch $\pi$ of the tree, the formula $\gamma$ holds at that position on all the other branches that pass by that node. And, because $T$ is a tree, the branches $\pi'$ that pass by this node, have the same prefix up to that position as $\pi$ and then the extended runs compatible with $\pi'$ have the same observations up to there.

\begin{lemma}\label{th:tild_form}
If a $\Sigma_1 \times 2^{\KKK}$-labelled $\OOO$-tree $T$ is \textit{fair} with respect to $\phi$, then
\begin{align*}
\forall \psi \in sub(\phi)&, \forall \pi \text{ branch of } T, \forall j, \forall \rho \in \R{T,\pi,j},  \\
&\text{ if } \R{T,\pi,j}, \rho,j \models \psi, \\
&\text{ then }  \R{T,\pi,j}, \rho,j \models \psi^{-1}
\end{align*}
where $\psi^{-1}$ is obtained from $\psi$ by replacing the atomic propositions $k_{\gamma}$ with the formulas $K \gamma$.
\end{lemma}
\vspace*{-7pt}
\begin{proof}
Let fix a branch $\pi = (a_0,K_0) o_0 (a_1, K_1) o_1 (a_2, K_2) o_2 ...$ of $T$. Then, the proof is by induction on the structure of $\psi$.

\begin{itemize}
\item if $\psi = p \in \Prop$, $\psi^{-1} = \psi$ and then, $\R{T,\pi,j}, \rho, j \models p$. 

\item if $\psi = k_{\gamma}$, since $\R{T,\pi,j}, \rho, j \models k_{\gamma}$, $k_{\gamma} \in \proj_2(\rho[j]) =  K_j$. 
 Then, because $k_{\gamma} \in sub(\psi)$, we have $k_{\gamma} \in K_j \cap sub(\psi)$.

Because $T$ is a fair tree with respect to $\phi$ and $\psi \in sub(\phi)$, $\forall \pi' \in T$ s.t.$\pi[0...j] = \pi'[0...j]$, for all runs $ \rho' \in \R{\pi'}$, we have that $\R{T},\rho',j \models \gamma$. Then, since $\R{T,\pi,j} =  \bigcup_{\pi' \text{ s.t. } \pi[0...j] = \pi'[0...j]} \R{\pi'}$,  we have $\R{T,\pi,j}, r, j \models K\gamma$.

\item if $\psi = \psi_1 \wedge \psi_2$, $\R{T,\pi,j}, \rho, j \models \psi_1 \wedge \psi_2$. That is,  $\R{T,\pi,j}, \rho, j \models \psi_1$ and $\R{T,\pi,j}, \rho, j \models \psi_2$. From the induction, we have $\R{T,\pi,j}, \rho, j \models \psi^{-1}_1$ and $\R{T,\pi,j}, \rho, j \models \psi^{-1}_2$ which means that $\R{T,\pi,j}, \rho, j \models \psi^{-1}_1 \wedge \psi^{-1}_2$. The proof is the same for $\psi = \psi_1 \vee \psi_2$.

\item if $\psi = \psi_1 \mathcal{U} \psi_2$,
$\R{T,\pi,j}, \rho, j \models \psi_1 \mathcal{U} \psi_2$. This means that $\exists j' \geq j$ s.t. $\forall j \leq k < j'$,  $\R{T,\pi,j}, \rho, k \models \psi_1$ and $\R{T,\pi,j}, \rho, j' \models \psi_2$.

From the fact that $\R{T,\pi,k} \subseteq R{T,\pi,j}$ and $\R{T,\pi,j'} \subseteq \R{T,\pi,j}$, we have that $\forall j\leq k < j'$,  $\R{T,\pi,k}, \rho, k \models \psi_1$ and $\R{T,\pi,j'}, \rho, j' \models \psi_2$.

 Then, from the inductive hypothesis, 
 $\forall j\leq k < j'$,  $\R{T,\pi,k}, \rho, k \models \psi^{-1}_1$ and $\R{T,\pi,j'}, \rho, j' \models \psi^{-1}_2$ 
 which means that
 $\forall j\leq k < j'$,  $\R{T,\pi,j}, \rho, k \models \psi^{-1}_1$ and $\R{T,\pi,j}, \rho, j' \models \psi^{-1}_2$. 
 This is because $\R{T,\pi,j}$ is the set of runs that are consistent with $\pi$ from position $0$ to position $j$ and $\R{T,\pi,j'}$ and $\R{T,\pi,k}$ are subsets of $\R{T,\pi,j}$(since $j \leq k < j'$). 
 Furthermore, all the runs in $\R{T,\pi,j'}$ (and $\R{T,\pi,k}$ respectively) are distinguishable from the rest of words in $\R{T,\pi,j}$. 
 Then, $\R{T,\pi,j}, \rho, j \models \psi^{-1}_1 \mathcal{U} \psi^{-1}_2$.

\item if $\psi = \bigcirc \psi_1$, $\R{T,\pi,j}, \rho, j \models \bigcirc \psi_1$ $\R{T,\pi,j}, \rho, j+1 \models  \psi_1$. 
Because $\R{T,\pi,j+1} \subseteq \R{T,\pi,j}$ and from the induction hypothesis, we have $\R{T,\pi,j+1}, \rho, j+1 \models \psi^{-1}_1$. Using the same argument as before, $\R{T,\pi,j}, \rho, j+1 \models \psi^{-1}_1$ and then $\R{T,\pi,j}, \rho, j \models \bigcirc \psi^{-1}_1$.

\item if $\psi = \Box \psi_1$, $\R{T,\pi,j}, \rho, j \models \Box \psi_1$. Then,$\forall k \geq j$, $\R{T,\pi,j}, \rho, k \models \psi_1$. 
Because $\forall k \geq j$ $\R{T,\pi,k} \subseteq \R{T,\pi,j}$ and from the inductive hypothesis, $\forall k \geq j$, $\R{T,\pi,k}, \rho, k \models \psi^{-1}_1$.
 Using the same argument as before, we get $\R{T,\pi,j}, \rho, k \models \psi^{-1}_1$, $ \forall k\geq j$.
 That is, $\R{T,\pi,j}, \rho, j \models \Box \psi^{-1}_1$

\item if $\psi = K \psi_1$, the fact that $\R{T,\pi,j}, \rho, j \models K \psi_1$ means that $\forall \rho' \in \R{T,\pi,j}$, $\R{T,\pi,j}, \rho', j \models \psi_1$. From the induction hypothesis applied for $\rho'$, we have that $\forall \rho' \in \R{T,\pi,j}$, $\R{T,\pi,j}, \rho', j \models \psi^{-1}_1$. That is, $\R{T,\pi,j}, \rho, j \models K \psi^{-1}_1$.

\item if $\psi = \neg \psi_1$, then $K\gamma$ is not a subformula of $\psi_1$ because $K\gamma$ does not occur under negations in $\psi$. This means that $\psi_1 = \psi^{-1}_1$ and then $\R{T,\pi,j}, \rho, j \models \neg \psi_1$ implies that $\R{T,\pi,j}, \rho, j \models \neg \psi^{-1}_1$.
\end{itemize}

\end{proof}

Then, denoting by $\RX X \pi$ the set of $e$-executions in $\modelEnv$ that are compatible with $\pi$ and start when the current possible states of $\modelEnv$ are $X$, we can prove that:

\begin{lemma}\label{th:fair}
For all  $i \in \{ 0,...,d-1 \}$, $\forall T \in \mathcal{L}(\TTT^i)$, $T$ is fair with respect to $\varphi^{i+1}$. 
\end{lemma}
\begin{proof}
 The proof of this theorem comes directly from the construction of $\TTT^{i}$ from $\TTT^{i+1}$.
 
 Let $ T \in \mathcal{L}(\TTT^{i})$. By the construction, $\TTT^{i}$ is $\TTT^{i+1}$ to which is "plugged" the automata $\TTT_{\gamma,X}$ for the LTL formulas $\gamma$ where $k_{\gamma} \in sub(\varphi^{i+1})$ appears on a transition of $\TTT^{i+1}$.
 
 Therefore, if at position $j$ on a branch $\pi = (a_0,K_0) o_0 (a_1,K_1) o_1 ...$ of $T$ we have $k_{\gamma} \in K_j \cap sub(\varphi^{i+1})$, there also starts the execution of $\TTT_{\gamma,X}$ that accepts the "maximal" subtree of $T$ with the root at position $j$ on branch $\pi$. 
 Let denote this subtree $T_{\pi,j}$ and observe that it is formed from the suffixes $\pi'[j...]$ of all the branches $\pi'$ of T such that $\pi'[0...j] = \pi[0...j]$. 
 Therefore, since $T_{\pi,j} \in \mathcal{L}(\TTT_{\gamma,X})$ for some $X$ and $\gamma$ is an LTL formula, 
using the results in Section 4 we have that $\RX X {\pi'[j...]} \models \gamma $.
 That is, $\forall \pi'$ branch of $T$ such that $\pi'[0...j] = \pi[0...j]$, $\forall \rho \in \R{ \pi'}$, we have $\R{\pi'},\rho,j \models \gamma$. 
 This means that $T$ is fair with respect to $\varphi^{i+1}$.
 
\end{proof}

\begin{proposition}\label{th:Ti_satisf}
For all $i \in \{0,...,d\}$, $\forall T \in \mathcal{L}(\TTT^i)$, $\forall \pi$ branch of $T$, $\forall \rho \in \R{\pi}$,
$$\R{\pi},\rho,0 \models \varphi^i$$ 
\end{proposition}
\begin{proof}
We do the proof of the theorem by induction on $i$. 
For the base case, if $i=d$, $\TTT^d = \TTT_{\varphi^d,S_0}$ and 
a sample adaptation of the proof of Lemma \ref{lemma:LTLReal}  
yields that $\forall \pi$ a branch of $T \in \lang(\TTT^d)$, we have directly that $\R{\pi} \models \varphi^d$.

For the inductive step, we suppose that the the theorem holds for $\TTT^{i+1}$ and prove it for $\TTT^i$.
From the inductive hypothesis, we have that $\forall T \in \mathcal{L}(\TTT^{i+1})$, $\forall \pi$ branch of $T$, $\forall \rho \in \R{\pi}$, $\R{\pi},\rho,0 \models \varphi^{i+1}$.

But, since $\mathcal{L}(\TTT^{i+1}) \supseteq \mathcal{L}(\TTT^i)$, it is true that $\forall T \in \mathcal{L}(\TTT^{i})$, $\forall \pi$ branch of $T$, $\forall \rho \in \R{\pi}$, we have $\R{\pi},\rho,0 \models \varphi^{i+1}$.

Let fix $T \in \mathcal{L}(\TTT^i)$ and a branch $\pi$ of $T$.
By Theorem \ref{th:fair}, $T$ is fair with respect to $\varphi^{i+1}$. Then, since $\R{\pi} \subseteq \R{T,\pi,0}$ and since all the runs $\rho \in \R{\pi}$ are distinguishable from the other runs from $\R{T,\pi,0} \setminus \R{\pi}$ (because they are compatible with different sequences $\pi'$ of actions $\tilde{a}_i$ and observations $o_i$), we have that $\forall \rho \in \R{\pi}$, $\R{T,\pi,0},\rho,0 \models \varphi^{i+1}$.

Then, by applying Lemma \ref{th:tild_form}, results that $\R{T,\pi,0},\rho,0 \models \varphi^{i}$ and because of the inclusion 
$\R{\pi} \subseteq \R{T,\pi,0}$ and because the runs $\rho \in \R{\pi}$ are distinguishable from the other runs from $\R{T,\pi,0} \setminus \R{\pi}$ ,
 we have that 
$\R{\pi},\rho,0 \models \varphi^i$.
\end{proof}

\vspace*{-5mm}
\section{Proof of corollary \ref{coro:realiz}}
\vspace*{-3mm}

\begin{proof}
From left to right, if $\lang(\TTT^0) \not = \emptyset$, 
there exists a tree $T \in \lang(\TTT^0)$ that, by Theorem \ref{th:Ti_satisf}, 
satisfies the property that $\forall \pi = (a_0,K_0) o_0 (a_1,K_1) o_1 ...$ branch of $T$, $\forall r \in \R{\pi}$,
$\R{\pi},r,0 \models \varphi^0$.  
This means that all paths in the model $\modelEnv$ compatible with the branches of the tree $T$ 
satisfy the formula $\varphi^0=\varphi$. 
That is, since $\varphi$ is a $K$-positive KLTL formula over $2^{\Prop}$, 
$\forall r \in \R{T}$, we have that $\proj_1(\R{T}),\proj_1(r),0 \models \varphi$.

Then, there exists a strategy $\lambda_1 : (\Sigma_1 \times \OOO)^* \rightarrow \Sigma_1$ of the system, 
represented by the tree $T$, that, 
since $\exec (\modelEnv, \lambda_1) = \bigcup_{\pi \in T} \proj_1(\R{\pi})$, 
satisfies the property that 
$\forall \rho \in \exec (\modelEnv, \lambda_1)$, 
we have $\exec (\modelEnv, \lambda_1), \rho, 0 \models \varphi$. 
That is, $\varphi$ is realizable in $\modelEnv$.

From right to left, if $\varphi$ is realizable in $\modelEnv$, 
there exists a strategy $\lambda_1 : (\Sigma_1 \times \OOO)^* \rightarrow \Sigma_1$ of the system
such that $\forall \rho \in \exec (\modelEnv, \lambda_1)$, we have 
$\exec (\modelEnv, \lambda_1), \rho, 0 \models \varphi$. 
This strategy can be seen as a $\Sigma_1$-labelled $\OOO$-tree T. 
Note that any branch $a_0 \xrightarrow{o_0} a_1 \xrightarrow{o_1} a_2 ...$ of $T$ 
can be seen as a sequence $\hat{\pi} = a_0 o_0 a_1 o_1 a_2 o_2 ... $ 
where $a_i \in \Sigma_1 \times \emptyset$ and $o_i \in \OOO$.
Therefore, $\forall r \in \R{\hat{\pi}}$, we have $\R{\hat{\pi}},r,0 \models \varphi$.

Now, we annotate all the nodes $a_j$ of the tree $T$ with fresh atomic propositions $k_{\gamma} \in \KKK$
whenever on the branch $\hat{\pi}$, the formula $K \gamma$ has to be true at position $j$, 
i.e., whenever for all the branches $\hat{\pi'}$ of $T$ such that 
$\hat{\pi}[0...j] = \hat{\pi'}[0...j]$, $\forall r \in \R{\hat{\pi'}}$, 
it is true that $\R{\hat{\pi'}},r,j \models \gamma$.
Then, all branches of the obtained tree $\tilde{T}$ will be of the form 
$\pi = (a_0,K_0) o_0 (a_1,K_1) o_1 (a_2,K_2) o_2 ...$ 
where $K_j = \{ k_{\gamma} \mid K \gamma \text{ is true at position j} \}$. 
Also, $\forall r \in \R{\pi}$, $\R{\pi},r, 0 \models \varphi^d$ thanks to the annotations. 
Therefore, the annotated tree $\tilde{T}$ is accepted by $\TTT^d$ and then, 
by the construction of the chain of over-approximations 
and by the fact that $k_{\gamma}$ appears only where $K \gamma$ is true, 
we have that $\tilde{T}$ is also accepted by $\TTT^0$.
This means that $\lang(\TTT^0) \not = \emptyset$.

\end{proof}

\vspace*{-5mm}
\section{Reduction to Safety Games and Complexity}
\vspace*{-2mm}

In the following, the aim is to reduce the emptiness problem to a safety game between the environment and the system.
Following the approach in \cite{CAV09}, we turn the universal Co-B\"{u}chi tree automaton $\TTT^0$ into an universal B-Co-B\"{u}chi tree automaton $(\TTT^0,B)$(in which at most $B$ accepting states are visited) and check for the emptiness. For doing this, we construct a two-player game $\mathcal{G_\varphi}$ with a safety winning condition. This is done via a determinization of the universal $B$-Co-B\"{u}chi tree automaton.

\vspace*{10pt}
In order to simplify the notations, in the next sections we will use $\TTT = \langle Q,Q_0,\Delta, \alpha, \Sigma_1 \times 2^{\mathbb{K}} \rangle$ instead of $\TTT^0$ since the others automata in the chain of over-approximations are not needed in the following.

\vspace*{-2mm}
\subsection{Reduction to Universal $B$-Co-B\"{u}chi tree automaton }

Before reducing the universal co-B\"{u}chi tree automaton $\TTT$ to a universal $B$-co-B\"{u}chi tree automaton, we mention the fact that a finite-state strategy can be represented by a Moore machine as in \cite{CAV09} where the transition relation is extended to simulate a strategy.

\begin{lemma}\label{lemma:preserv}
Let $\TTT$ be a UCT over $\Sigma$ with n states constructed for the \pKLTL formula $\varphi$ and a strategy $\lambda$ represented by a Moore machine $M_\lambda$ with m states. Then, $T_\lambda \in \mathcal{L}_{uc}(\TTT)$ iff $T_\lambda \in \mathcal{L}_{uc,2nm}(\TTT)$, where the strategy $\lambda$ is viewed as the tree $T_\lambda$.  
\end{lemma} 
\begin{proof}
It is obvious that if $T_\lambda \in \mathcal{L}_{uc,2nm}(\TTT)$, then  $T_\lambda \in \mathcal{L}_{uc}(\TTT)$.
Now, if $T_\lambda \in \mathcal{L}_{uc}(\TTT)$, intuitively, the infinite runs in $\TTT$ on $T_\lambda$ are accepting. Thar is, each path of the runs on $T_\lambda$ visits finitely many accepting states.
Therefore, in the product of $M_\lambda$ with $\TTT$, there are no cycle visiting accepting states of $\TTT$ which bounds the number of visited accepting states on a path by the number of states in the product. 
\end{proof}

Further, in \cite{KuVa05} is shown that if a UCT automaton with n states is not empty, then it accepts a finite state machine of width bounded by $(2n!)n^{2n}3^n (n+1)/n!$. An improved upper bound result mentioned in \cite{KuPi09} shows that the width reduces to $2n(n!)^2$. 
Then, using this results we can turn the emptiness problem for the UCT automaton $\TTT$ into the emptiness problem of the UBCB automaton $(\TTT, B)$ as follows:

\begin{theorem} \label{th:bound}
Given the UCT $\TTT$ over $\Sigma_1 \times 2^{\mathbb{K}}$ with n states and $B=4n^2(n!)^2$, $\mathcal{L}_{uc}(\TTT) \not = \emptyset$ iff $\mathcal{L}_{uc,B}(\TTT) \not = \emptyset$.
\end{theorem}
\begin{proof}
If $\mathcal{L}_{uc}(\TTT) \not = \emptyset$, then there exists a regular tree $T_\lambda \in \mathcal{L}_{uc}(\TTT)$ generated by a finite state machine with $m$ states($m<2n(n!)^2$). Then, by the Lemma \ref{lemma:preserv}, $T_\lambda \in \mathcal{L}_{uc,2nm}(\TTT)$ and $\mathcal{L}_{uc,B}(\TTT) \not = \emptyset$. In the other sense, the proof is obvious since $\mathcal{L}_{uc,B}(\TTT) \subseteq \mathcal{L}_{uc}(\TTT)$.

\end{proof}

\vspace*{-5mm}
\subsection{Reduction to Safety Game} 

In the previous subsection we reduced the emptiness problem of the universal Co-B\"{u}chi tree automaton $\TTT$ to the emptiness problem of the universal $B$-Co-Buchi tree automaton $(\TTT,B)$. Further, we show the reduction of the new problem to a safety game. This is done via an determinization of the universal $B$-Co-B\"{u}chi tree automaton $(\TTT,B)$. 
Note that we explain the reduction for self-containess reasons, the construction being very similar to \cite{Filiot11,ScheweF07a}.

In the determinization step, we construct a complete deterministic 0-Co-B\"{u}chi tree automaton by extending the subset construction with counters. The states of the automaton are functions $F$ that count(up to $B+1$) for each state $q \in Q$ the maximum number of accepting states visited by the paths that lead to $q$. Formally, $F:Q \rightarrow \{ -1,0,...B+1 \}$ where $F(q) = -1$ means that no runs on the prefix read so far end in $q$. Also, $0,...,B$ are safe states for which $F(q)=i,0 \leq i \leq B$ means that the maximal number of visits of final states by runs that end in $q$ is $i$, and $B+1$ is the unsafe state where the number of visits to final states of runs that end in $q$ is greater or equal to $B+1$.

Then, the determinization of $(\TTT,B)$ is the universal 0-co-B\"{u}chi tree automaton $Det(\TTT,B) = \{ \mathcal{F}, F_0, \alpha', \delta \}$ where $\mathcal{F}$ is the set of states, $F_0$ is the initial state, $\alpha'$ is the set of final states and $\delta$ is the transition relation with:
\begin{itemize}
\item  $\FFF = \{ F | F:Q \rightarrow \{ -1,0,...,B+1 \} \} $
\item  $\forall (q,I) \in Q, F_0(q,I) = -1$ if $(q,I) \not \in Q_0$ and $F_0(q,I) = ((q,I) \in \alpha)$ otherwise
\item  $\alpha' = \{ F \in \FFF | \exists (q,I)\in Q \text{ s.t. } F(q,I)>B \}$
\item  $\forall o \in \OOO$, $\delta(F,(a,K),o) = (F',o)$ if\\ 
	   \hspace*{10pt}  $F'(q',I') = max \{ min \{ B+1, F(q,I)+ ((q',I')\in \alpha) \mid$\\
       \hspace*{80pt} $ ((q',I'),o) \in \Delta((q,I),(a,K),o) \text{ and } F(q,I) \not = -1 \} \}$
\end{itemize}
where $max(\emptyset)=-1$ and, for all states $(q,I) \in Q$, $((q,I) \in \alpha)=1$ if $(q,I)$ is in $\alpha$ and $0$ otherwise. We say that a state $F$ is unsafe if there exists $(q,I) \in Q$ such that $F(q,I)=B+1$.

\begin{theorem}\label{th:det}
Let $\TTT$ be the UCT for for the $\pKLTL$ formula $\varphi$ and $B \in \mathbb{N}$. Then, $Det(\TTT,B)$ is complete, deterministic and $\mathcal{L}_{uc,0}(Det(\TTT,B)) = \mathcal{L}_{uc,B}(\TTT)$.
\end{theorem}
\begin{proof} It is obvious that the constructed automaton is complete by construction 
and from the completeness of $(\TTT,B)$. Also, it is deterministic because of the construction. 
From a state $F$, for one action of the agent and one observation we get only one successor state $F'$.

Then, we prove the language equality by double inclusion.
From right to left, if $\langle T,\tau \rangle$ is a 
$\Sigma_1 \times 2^{\mathbb{K}}{-}labelled$ $\OOO{-}tree$ accepted by $(\TTT,B)$, 
there exists an accepting run on $(T,\tau)$ in $(\TTT,B)$. 
Let it be $\langle T_r,\tau_r \rangle$ such that 
 $\tau_r(\epsilon) \in Q_0 \times \{ \epsilon \}$ and
 $\forall x \in T_r$ s.t. $\tau_r(x) = ((q,I),v)$ and $\tau(v) = (a,K)$,
 if $\Delta((q,I),(a,K))= \{ ((q_1,I_1),o_1),..., ((q_n,I_n),o_n) \}$, 
 then $\forall 0 < j \leq n$, $xj \in T_r$ and $\tau_r(xj) = ((q_j,I_j),v \cdot o_j)$.

Since the run on $T$ is accepting, each branch $\pi$ of $t_\lambda$ induces a infinite sequence 
$\tau_r(\pi) = ((q_0, I_0),\epsilon) ((q_1, I_1),v_1) ((q_2, I_2),v_2)...$
which visits at most $B$ times the final states.
Then, by the construction of $Det(\TTT,B)$ using subset construction on the states with the same observation,
there exists a run on $\langle T,\tau \rangle$ 
which is accepting since the set of final state 
$\alpha' = \{ F \in \mathcal{F} | \exists (q,I)\in Q' \text{ s.t. } F(q,I)>B \}$ 
it is not reached.

From left to right, let $\langle T,\tau \rangle$ be a 
$\Sigma_1 \times 2^{\mathbb{K}}{-}labelled$ $\OOO{-}tree$ accepted by $Det(\TTT,B)$. 
Then, there exists a run on  $\langle T,\tau \rangle$ in $Det(\TTT,B)$. 
Let it be $\langle T_r,\tau_r \rangle$ such that
 $\tau_r(\epsilon) \in F_0 \times \{ \epsilon \}$ and
 $\forall x \in T_r$ s.t. $\tau_r(x) = (F,v)$ and $\tau(v) = (a,K)$,
 if $\delta(F,(a,K))= \{ (F_1,o_1),..., (F_n,o_n) \}$, 
 then $\forall 0 < j \leq n$, $xj \in T_r$ and $\tau_r(xj) = (F_j,v \cdot o_j)$.

Since the run is accepting in the 0-Co-B\"{u}chi tree automaton, it means that 
it doesn't visit any $F$ such that $\exists (q,I)$ for which $F(q,I) \geq B+1$. 
This means that, using the construction of $Det(\FFF,B)$, 
the branches of the run in $(\TTT,B)$ on $\langle T,\tau \rangle$ visit at most $B$ times the final states in $\alpha$.
Therefore, $T \in \lang_{uc,B}(\TTT)$.

\end{proof}

A two-player safety game is played on a game arena. 
To reflect the game point of view, the automaton $Det(\TTT,B)$ 
can be seen as a \textit{game arena} $G(\TTT,B)=(\FFF_E,\FFF,q_0,T=T_1 \cup T_E,Safe)$ 
with $|\FFF| + |\delta|$ states where 
$\FFF_E \cup \mathcal{F}$ is the set of states of the game,
$\FFF$ being the set of states controlled by the system and 
$\mathcal{F}_E$ the set of states controlled by the environment, 
$q_0 = F_0 \in \mathcal{F}$ is the initial state,
$T \subseteq \mathcal{F} \times \mathcal{F}_E \cup \mathcal{F}_E \times \mathcal{F}$ is the transition relation and 
$Safe \subseteq \mathcal{F}$ is the safety winning condition defined by 
\begin{itemize}
\item  $\mathcal{F}_E = \{ F_{a,K} | F \in \FFF \text{ and } (a,K) \in \Sigma_1 \times 2^\KKK \}$
\item  $T_1 = \{ (F, F_{a,K}) | F \in \FFF \text{ and } (a,K) \in \Sigma_1 \times 2^\KKK \}$
\item  $T_E = \{  (F_{a,K},F') | \forall F' \in \delta(F,(a,K)) \}$ 
\item  $Safe = \mathcal{F} \setminus \alpha'$
\end{itemize}
Then, $\FFF_E \cup \mathcal{F}$ is the total set of states of the game and $T = T_1 \cup T_E$ is the transition relation.

 A infinite \textit{play} on $G(\TTT,B)$ is a path $\rho = \rho_0 \rho_1 \rho_2 ... \in (\mathcal{F}\mathcal{F}_E)^\omega$ such that $\forall i \geq 0$, $(\rho_i,\rho_i+1) \in T$. Finite plays are similarly defined and they belong to $(\mathcal{F}\mathcal{F}_E)^*$. 
 A strategy for the system is a mapping $\lambda_1:(\mathcal{F}\mathcal{F}_E)^*\mathcal{F} \rightarrow \mathcal{F}_E$ that maps every finite play $\rho$ whose last state is $\rho_n \in \mathcal{F}$ to a state $\rho_{n+1}$ such that $(\rho_n,\rho_{n+1}) \in T_1$. A strategy for the environment is a mapping $\lambda_2:(\mathcal{F}\mathcal{F}_E)^* \rightarrow \mathcal{F}$ that maps an finite play ending in $\rho_n \in \mathcal{F}_E$ to a state $\rho_{n+1} \in \mathcal{F}$ such that $(\rho_n,\rho_{n+1}) \in T_E$.
 The \textit{outcome} of a strategy $\lambda_1$ of the system is a set $Outcome_{G(\TTT,B)}(\lambda_1)$ of infinite plays $\rho = \rho_0 \rho_1 \rho_2 ... \in (\mathcal{F} \mathcal{F}_E)^\omega$ such that if $\rho_i \in \mathcal{F}$, then $\rho_{i+1} = \lambda_1(\rho_0 ... \rho_i)$.
 A strategy for the system is winning if $Outcome_{G(\TTT,B)}(\lambda_1) \subseteq (Safe \; \mathcal{F}_E)^\omega$.
 
 \begin{theorem}\label{th:win_G}
 $\mathcal{L}_{uc,0}(Det(\TTT,B)) \not = \emptyset$ iff the system has a winning strategy $\lambda_1$ in the game $G(\TTT,B)$.
 \end{theorem}  
 \begin{proof}
 If $\mathcal{L}_{uc,0}(Det(\TTT,B)) \not = \emptyset$, then there exists a $\Sigma_1 \times 2^{\mathbb{K}}{-}labelled$ $\OOO{-}tree$ $T$ which is accepted by $Det(\TTT,B)$. Then, there exists an accepting run $\langle T_r,\tau_r \rangle$ on $T$ whose paths don't visit the final states in $\alpha'$.
 Then, by the construction of the game $G(\TTT,B)$, there exists a strategy(represented by the tree $T$) which is winning in $G(\TTT,B)$ because all the paths in the game arena that follow it visit exactly the same states in $\mathcal{F}$ as $\langle T_r,\tau_r \rangle$ does and the safe set is defined as $Safe = \mathcal{F} \setminus \alpha'$.
 
 Now, if there exists a winning strategy $\lambda_1$ in $G(\TTT,B)$(which can be seen as a $\Sigma_1 \times 2^{\mathbb{K}}{-}labelled$ $\OOO{-}tree$ $T$), then all the outcomes of $\lambda_1$ will stay in the set $Safe$. Then, by construction of the game arena, there exists a run on $T$ in the automaton $Det(\TTT,B)$ which is also accepting because $Safe = \mathcal{F} \setminus \alpha'$ which means that the run doesn't visit the set $\alpha'$.
 
 \end{proof}

Therefore, by Theorems \ref{thm:main}, \ref{th:bound}, \ref{th:det} and \ref{th:win_G}, we have that the \pKLTL formula $\varphi$ is realizable in the environment model $\modelEnv$ if and only if the system has a winning strategy in the game  $G(\TTT,B)$.

\vspace*{-2mm}
\subsection{Complexity} 

In the following, we study the complexity of the realizability algorithm for the \pKLTL formulas. First, using the way the game arena is constructed, we have the following lemma:

\begin{lemma}\label{lemma:knowledge}
For all states $F \in \mathcal{F}$, if there exist $(q,I), (q,I') \in Q$ such that $F(q,I) \not = -1$ and $F(q,I') \not = -1$, if $I \not = I'$, then $F$ is not reachable from the initial state.
\end{lemma}
\begin{proof}
We prove the lemma by induction on the length of the path leading to a state $F$.

If $F = F_0$, then since $F(q,I) \not = -1$, we have 
$(q,I) \in Q_0 = \{ (q_0 , S_0) \mid q_0 \in Q_\psi^0 \}$. 
Then, the proposition is true by definition.

Now, suppose that there is reached a state $F_i \in \mathcal{F}$ such that 
for all $(q,I),(q,I') \in Q$ with $F_i(q,I) \not = -1$ and $F_i(q,I') \not = -1$, 
we have $I=I'$. Let this set be denoted by $I_i$. 
We then prove that, after one round in the game, the state $F_{i+1}$ that is reached  
satisfy the property that for all $(q,I),(q,I') \in Q$ such that 
$F_{i+1}(q,I) \not = -1$ and $F_{i+1}(q,I') \not = -1$, we have $I=I'$.

From the construction of the game arena, $F_{i+1}$ is computed from $F_i$ 
using the action $(a,K)$ proposed by the agent and observation $o$ given by the environment. 
That is,
$F_{i+1}(q,I') = max \{ min \{ B+1, F_i(q'',I_i)+ ((q,I')\in \alpha) \mid ((q,I'),o) 
\in \Delta((q'',I_i),(a,K)) \text{ and } F_i(q'',I_i) \not = -1 \} \}$.
From the definition, the states $(q,I')$ for which $F_{i+1}(q,I') \not = -1$ 
are states for which there exists $(q'',I_i)$ with $F_i(q'',I_i) \not = -1$ and 
$I'=Post_a(I^i_{q'',q},o)$ 
Therefore, 
since the state $F_{i+1}$ is uniquely determined by a sequence of actions and observations, 
and the states $q$ and sets $I$ are synchronized along a path, 
we have that all the states $(q,I')$ with $F_{i+1}(q,I') \not = -1$ 
contain the same set $I_{i+1} = I'$ of states of the environment.

\end{proof}

Since we are interested in reachable positions of the safety game only, 
we can just focus on functions $F$ such that for all $(q,I), (q,I') \in Q$, we have $I=I'$. 
We associate to each function $F$ a tuple $\bar{K}$ of sets $I$ corresponding to the states $q$(as in section 6) 
and denote by $\bar{K}(q)$ the set $I$ associated to the state $q$ in $F$. 
Then, the safety game $G(\TTT,B)$ is restricted to reachable states $(F,\bar{K})$.

\begin{lemma}\label{prop:no_states}
The number of states of the safety game 
$G(\TTT,B)$  built for  the \pKLTL formula $\varphi$ in the environment $\modelEnv$
has at most 
$\frac{(B+3)^{2^{|S|+|\varphi|}+1}-1}{B+2} \times (1+ 2^{|S|}) \times (1 + | \Sigma_1 \times 2^{\mathbb{K}} |)$ states.
\end{lemma}
\begin{proof}
$\mathcal{F} = \{ F \mid F: Q \rightarrow \{ -1,0,...,B+1 \} \}$. But, according to the previous lemma, 
$\forall (q,I), (q',I') \in Q$ such that $F(q,I) \not = -1$ and $F(q',I') \not = -1$, we have $I = I'$. 
It means that, for a function $F$, there are at most 
$n= |Q| + \sum_{\gamma : \exists i \text{ s.t. } K \gamma \in Sub(\varphi^i) } |Q_\gamma|$ 
states $(q,I)$ for which $F(q,I) \not = -1$. 

Then, because of the way the set $I$ is computed using $Post_a(Y,o)$, excepting the initial state 
that may fall into several observations, there exists $o \in \OOO$ such that $I \subseteq o$.  
This means that there exist $1 + \sum_{o \in \OOO} 2^{|o|}$ sets of states of the environment 
that are reached, where $1$ is for the initial state. 

In conclusion, the number of functions F is 
$[(B+2)^0 + (B+2)^1+ ... (B+2)^n] \times (1 + \sum_{o \in \OOO} 2^{|o|})$. 
That is, $\frac{(B+3)^{n+1}-1}{B+2} \times (1 + \sum_{o \in \OOO} 2^{|o|})$ functions.

Then, for each state $F$, we have $| \Sigma_1 \times 2^{\mathbb{K}} |$ 
states $F_{a,K}$ controlled by the environment. 
This gives a total number of 
$\frac{(B+3)^{n+1}-1}{B+2} \times (1 + \sum_{o \in \OOO} 2^{|o|}) \times (1 + | \Sigma_1 \times 2^{\mathbb{K}} |)$ 
of states of the game  $G(\TTT,B)$.

Now, since $n$ can be bounded by $2^{|S|+|\varphi|}$ and $\sum_{o \in \OOO} 2^{|o|} \leq 2^|S|$, 
we get that the number of states of $G(\TTT, B)$ is bounded by 
$\frac{(B+3)^{2^{|S|+|\varphi|}+1}-1}{B+2} \times (1 + 2^|S|) \times (1 + | \Sigma_1 \times 2^{\mathbb{K}} |)$.
\end{proof}

\begin{proposition}\label{theor:complex}
The realizability of a K-positive KLTL formula $\varphi$ is decidable in 2EXPTIME.
\end{proposition}
\begin{proof} 
As also showed in Section 5, the number of states if the co-B\"{u}chi tree automaton $\TTT$ 
is bounded by $2^{|S|+|\varphi|}$ and is built in EXPTIME since there is a polynomial number of formulas $K \gamma$ in $\phi$.

Then, when transforming $\TTT$ into a $B$-co-B\"{u}chi tree automaton, the bound $B$ equals to $4n^2(n!)^2$. 
By applying Lemma \ref{prop:no_states}, we have that the safety game $G(\TTT,B)$ is built in 2EXPTIME.
\end{proof}

\vspace*{-5mm}
\section{Implementation and Case Studies}
\vspace*{-3mm}

\begin{figure}
\centering
\vspace*{-20pt}
\includegraphics[width=0.8\textwidth]{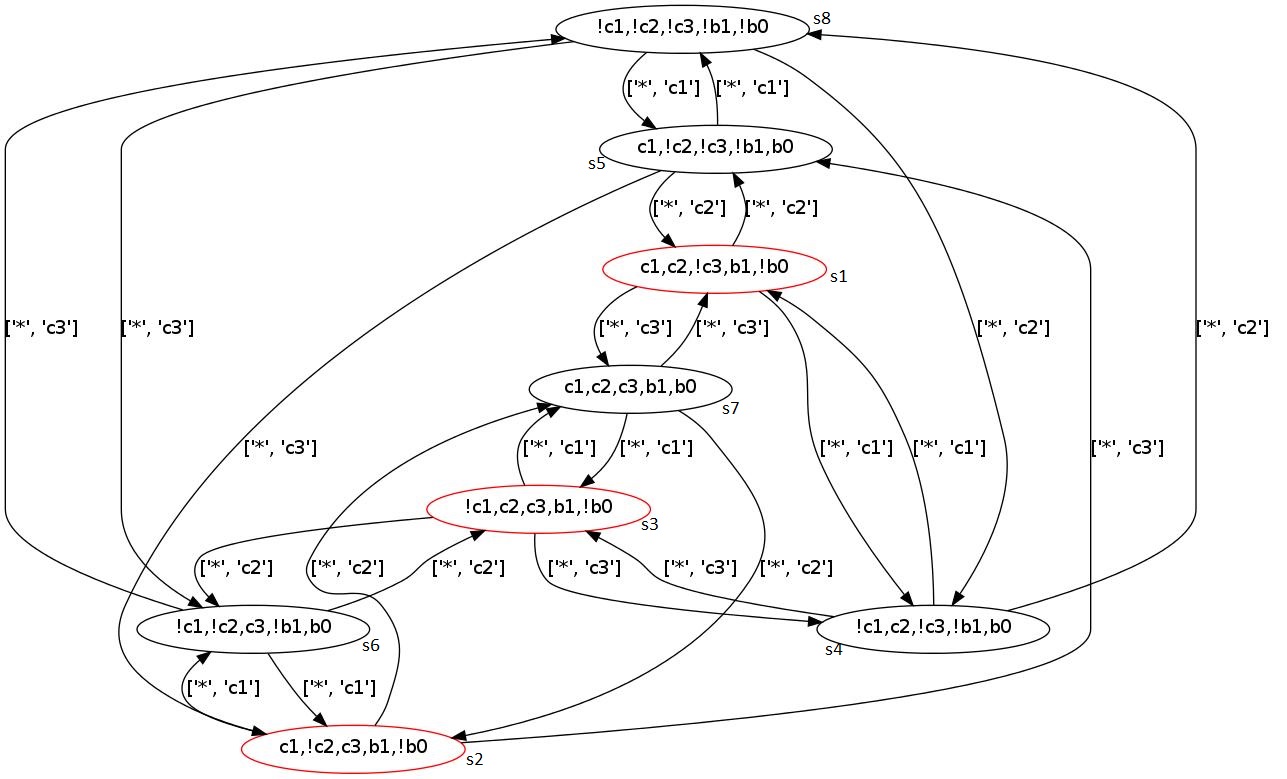}
\caption{The environment for 3-Coins Game}
\label{fig:env-coins}
\end{figure}

\begin{figure}
\centering
\vspace*{-20pt}
\includegraphics[width=0.8\textwidth]{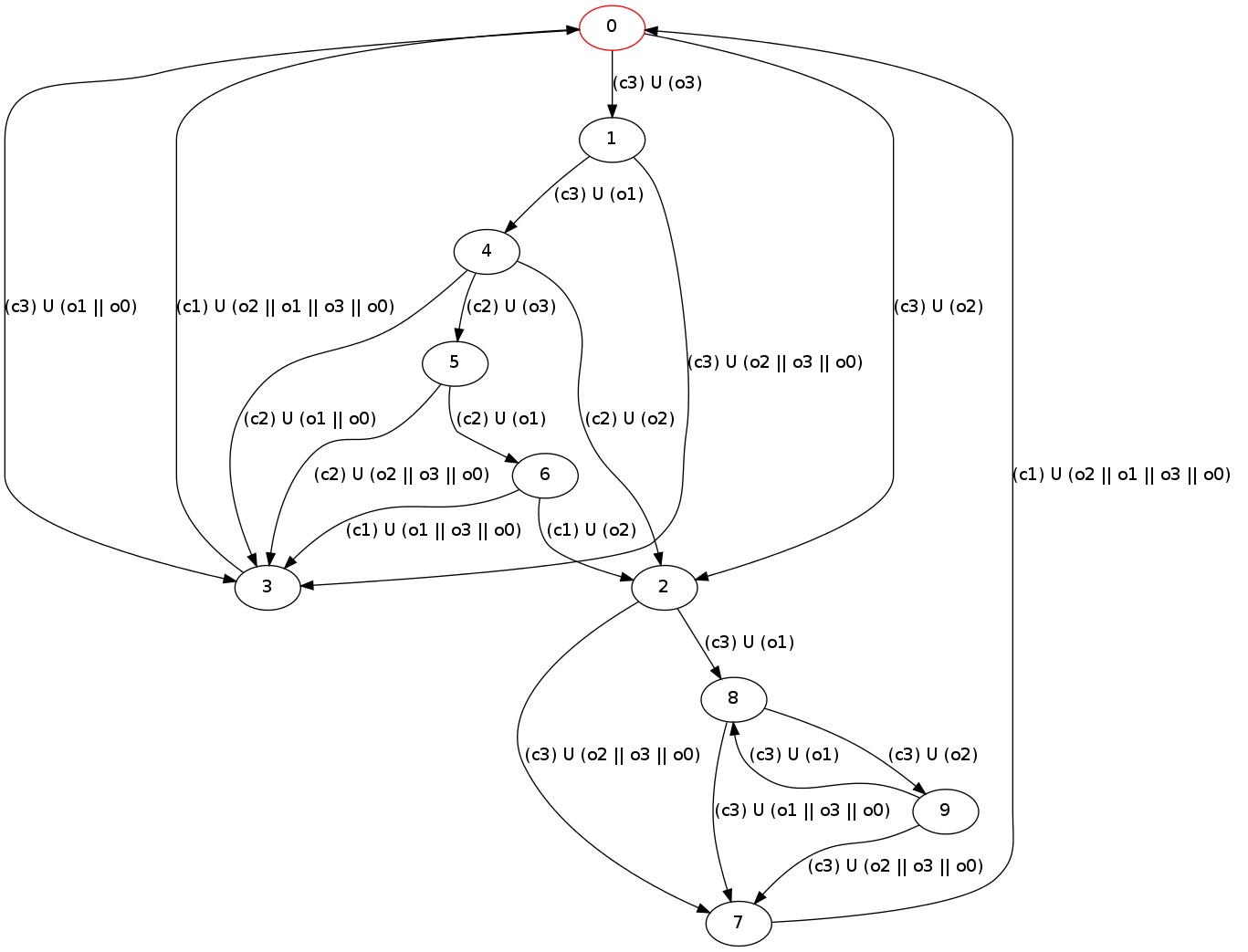}
\caption{The winning strategy for 3-Coins Game}
\label{fig:str-coins}
\end{figure}

In the Moore machine representing the strategy of Fig.\ref{fig:str-coins}, the state $2$ corresponds to the configuration in which there are only \textit{heads} and the states $3$ and $7$ are the states in which he goes when the environment cheated. Then, the strategy has two parts. One that leads to the state $2$ by checking the coins and the second part in which the system plays only $c_3$ and moves between the states $s_1$ and $s_7$.

Another example that illustrates the game with imperfect information
in which we need the knowledge is based on the following enigma:

\begin{example}[n-Prisoners Enigma]
Consider that there are $n$ prisoners in a prison, each one in his own cell and they cannot communicate. 
Also, there is a room with a light bulb and a switch and a policeman that, at each moment of time, sends only one prisoner in that room and gives him the possibility to turn on or off the light. 
The prisoners can observe only the light when they are in the
room. The guardians send the prisoners in the room in any order and
infinitely many times if the game never stops (fairness assumption). At any time, any prisoner can stop the game. At that
point, if every prisoner has visited the room at least once, then they
are all free. Otherwise they will all stay in the jail for
eternity. Of course all the prisoners want to be freed, and therefore
if someone stops the game, he must be sure that all the prisoners have
indeed visited the room at least once. Before the game starts, they
are allowed to communicate, and they know the initial state of the
light.

If you want to solve this puzzle by yourself, don't read the following
paragraph which gives the solution. Assume that the light is initially
off. The solution is that there is a special prisoner, let say
prisoner $n$, that will count  up to $n-1$. For all $1\leq j\leq n-1$,
the fairness assumption ensures that prisoner $j$ will visit the room
again and again until the game stops. The first time he visits the
room while the light is off, it turns it on, otherwise it does
nothing. Prisoner $n$ will turn the light off next time he enters the
room, and increment his counter by $1$. When the counter reaches
$n-1$, prisoner $n$ stops the game because he is sure that all the
prisoners have visited the room at least once.
\end{example}

To model this problem, it is natural to represent the guardians by the
environment and the prisoners by multi-agents. However, our
framework only allows for one agent. Therefore we fix the strategy of
prisoners $1$ to $n-1$ and encode them in the environment model. 
Prisoner $n$ (the system) must figure out a winning strategy (ideally
the counting strategy described above). We have modelled this example
in Acacia-K and indeed, our tool find the strategy described above.
Let us now give more details about the formalization.



For three prisoners, $\Prop = \{ on,x_1, x_2, p_1, p_2, p_3 \}$ where the atomic proposition 
$on$ corresponds to the light, values of $x_i$ for $i \in \{1,2\}$ is $true$ if the prisoner $i$ already turned the light on,
and the proposition $p_i$ for $i \in \{1,2,3\}$ indicates the prisoner that is inside the room.
Then, $\Propv = \{ on, p_3 \}$ and $\Propi = \{ x_1, x_2, p_1, p_2 \}$ indicating that the last prisoner 
sees all the time the light and can observe when he is inside the special room but cannot see what the other prisoners do. 

We assume that at the beginning there is no one in the room. 
Them, the environment can propose an action in $\Sigma_2 = \{ P_1, P_2, P_3 \}$ 
deciding which prisoner is going in the room and prisoner $3$ will decide if he wants the light on or off by choosing an action in $\Sigma_1 = \{ t_{on}, t_{off} \}$. Note that the action of $P3$ is ignored if he is not chosen by the environment. 
The transition relation asks that if the prisoner $p_i, i \in \{ 1,2 \}$ finds for the first time the light off ($x_1=false$ and $on=false$), he turns on the light and the value of $x_i$ changes and remains $true$
for all the reachable states from there.


Then, assuming that the environment is restricted to send all the prisoners in the special room infinitely many times, the 
\pKLTL formula that translates the goal is 
$\Box \wedge_{i=1}^{n} (\Diamond p_1)  \rightarrow \Diamond  K(\wedge_{i=1}^{n-1} x_i)$.
A winning strategy for the prisoner $n$ would be to turn off the light whenever he is sent to the special room and to let it off if it already is. Then, after he finds the light on $n-1$ times when he is sent in that room, thanks to the strategy of the other prisoners, he will know that all of them passed by that room, and even more, all of them switched an the light.
Assuming that the observations set $\OOO = \{ o_0, o_1, o_2, o_3 \}$ where 
$o_0 = \{ s \in S \mid  on \in \tau(s) \text{ and } p_3 \not \in \tau(s) \}$,
$o_1 = \{ s \in S \mid  on \not \in \tau(s) \text{ and } p_3 \in \tau(s) \}$,
$o_2 = \{ s \in S \mid  on \in \tau(s) \text{ and } p_3 \in \tau(s) \}$ and 
$o_3 = \{ s \in S \mid  on \not \in \tau(s) \text{ and } p_3 \not \in \tau(s) \}$,
the strategy synthesized by \Acacia-K for three prisoners and corresponds to the intuitive strategy is illustrated in Figure \ref{fig:str-pris} where the state $10$ in the generated Moore machine corresponds to the moment when the prisoner $p_3$ knows that all the other prisoners passed through the special room and turned on the light.

\begin{figure}
\centering
\vspace*{-10pt}
\includegraphics[width=1\textwidth]{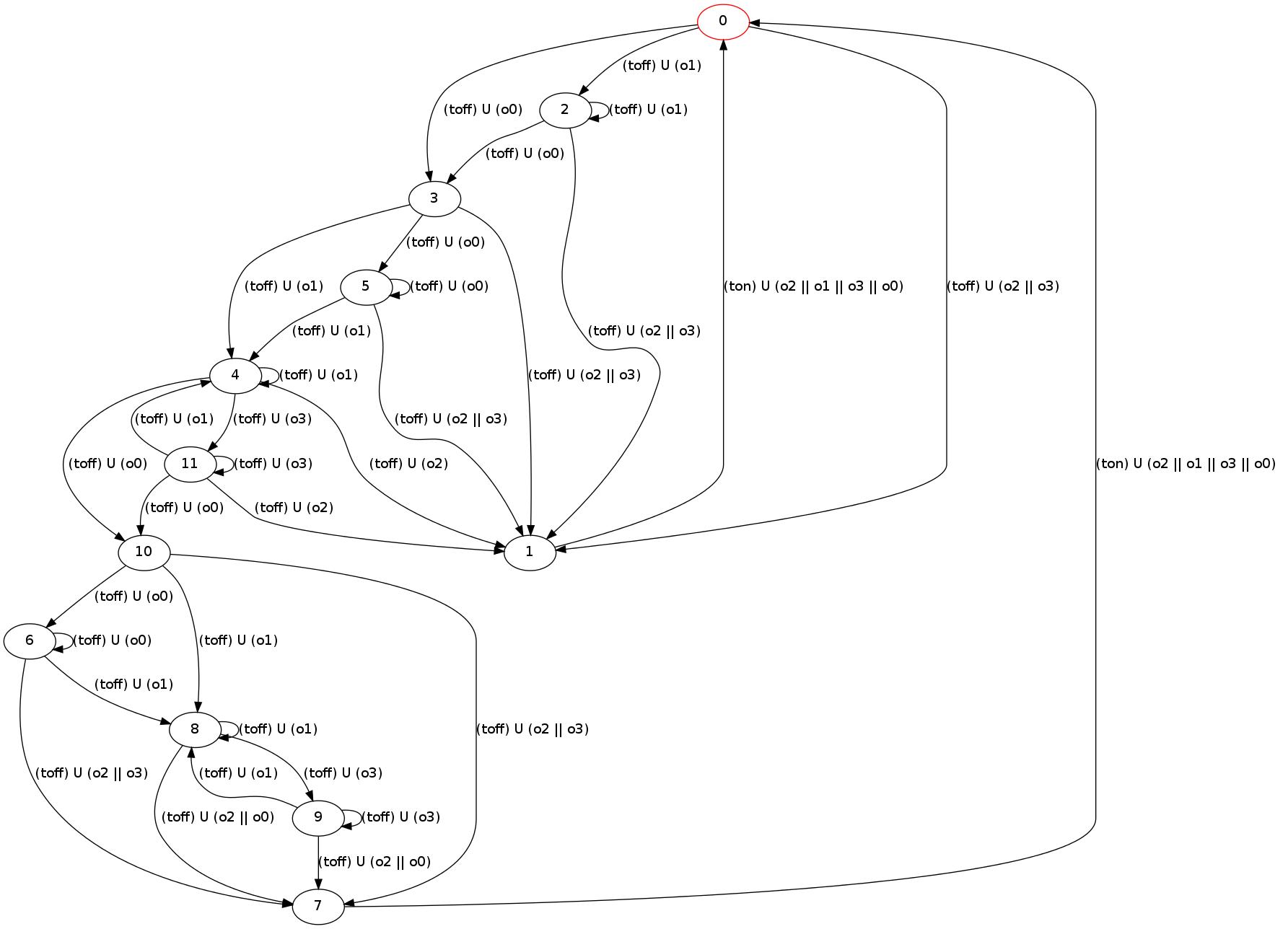}
\caption{The winning strategy for 3-Prisoners Game}
\label{fig:str-pris}
\end{figure}

For this example, \Acacia-K constructed a UCT with 144 states,
synthesised a strategy with 12 states, and the total running time is
1.87s.

\end{document}